\documentclass[a4paper]{article}
\usepackage[english]{babel}
\usepackage{amscd}
\usepackage{amsmath}
\usepackage{amsthm}
\usepackage{amssymb, faktor, titlesec}
\usepackage{graphicx, tikz-cd}
\usepackage[left=2.5cm,right=2.5cm, top=2.5cm,bottom=2.5cm,bindingoffset=0cm]{geometry}
\usepackage{enumerate, multicol, makecell}
\usepackage{setspace}
\titleformat{\section}{\large\bfseries\filcenter}{\thesection}{1em}{}
\titleformat{\subsection}{\bfseries}{\thesubsection}{1em}{}
\titleformat{\subsubsection}[runin]{\bfseries}{\thesubsubsection}{1em}{}[.]
\usepackage{enumitem}
\newlist{longenum}{enumerate}{5}
\setlist[longenum,1]{label=\alph*)}
\renewcommand{\emptyset}{\varnothing}

\usepackage[numbers, sort&compress]{natbib}

\theoremstyle{plain}

    \newtheorem{lemma}{Lemma}
    \newtheorem{theorem}{Theorem}
    \newtheorem{consequence}{Corollary}[section]
    \newtheorem{statement}{Proposition}[section]

\theoremstyle{remark}

  \theoremstyle{definition}
    \newtheorem{definition}{Definition}

            \newtheorem{remark}{Remark}[section]

\newcommand{\Ker}[1]{\mathrm{Ker} \, #1}

\newcommand{\rank}[1]{\mathrm{rank} \, #1}

\newcommand{\corank}[1]{\mathrm{corank} \, #1}

\newcommand{\Sing}[1]{\mathrm{Sing} \, #1}

\newcommand{\diff}[1]{\mathrm{d}  #1}
\newcommand{\diffFX}[2]{ \frac{\partial #1}{\partial #2} }

\newcommand{\diffFXi}[2]{ {\partial #1}/{\partial #2} }
\newcommand{\diffFXYp}[2]{ \frac{\mathrm{d} #1}{\mathrm{d} #2} }

\newcommand{\diffXp}[1]{ \frac{\mathrm{d} }{\diff #1} }

\newcommand{\Z}{\mathbb{Z}}

\newcommand{\R}{\mathbb{R}}
\newcommand{\Complex}{\mathbb{C}}
\newcommand{\Curve}{{C}}

\newcommand{\T}{\mathrm{T}}
\newcommand{\Cont}{\mathrm{C}}
\newcommand{\const}{\mathrm{const}}
\newcommand{\CP}{{\mathbb{C}}\mathrm{P}}
\newcommand{\CPP}{\overline \Complex}

\newcommand{\Hom}{\mathrm{H}}
\newcommand{\Ree}[1]{\mathrm{Re} \, #1}
\newcommand{\Imm}[1]{\mathrm{Im} \, #1}
\newcommand{\Tr}[1]{\mathrm{Tr} \, #1}
\newcommand{\tr}[1]{\mathrm{tr} \, #1}
\newcommand{\quo}[2]{\left.{#1}\middle/{#2}\right.}

\newcommand{\eps}{\varepsilon}
\newcommand{\wave}{\widetilde}

\newcommand{\g}{\mathfrak{g}}
\newcommand{\h}{\mathfrak{h}}

\newcommand{\so}{\mathfrak{so}}

\newcommand{\Pic}{\mathrm{Pic}}
\newcommand{\Jac}{\mathrm{Jac}}
\newcommand{\Div}{\mathrm{Div}}
\newcommand{\PDiv}{\mathrm{PDiv}}
\newcommand{\Res}{\mathrm{Res}}

\newcommand{\PGL}{\mathbb{P}\mathrm{GL}}

\newcommand{\E}{\mathrm{E}}

\newcommand{\sP}{\mathfrak{sp}}

\newcommand{\gl}{\mathfrak{gl}}

\renewcommand{\L}{  {L}}
\newcommand{\F}{{\pazocal F}}

\renewcommand{\deg}{\mathrm{deg}\,}
\usepackage{calrsfs}
\DeclareMathAlphabet{\pazocal}{OMS}{zplm}{m}{n}

\title{\Large Singularities of integrable systems and nodal curves}
\author{\normalsize Anton Izosimov\footnote{E-mail: izosimov@math.utoronto.ca}
}
\date{}
\begin{document}
\maketitle
%\tableofcontents %$J \in \gl(n,\Complex)$
\abstract{The relation between integrable systems and algebraic geometry is known since the XIXth century.
The modern approach is to represent an integrable system as a Lax equation with spectral parameter. In this approach, the integrals of the system turn out to be the coefficients of the characteristic polynomial $\chi$ of the Lax matrix, and the solutions are expressed in terms of theta functions related to the curve $\chi = 0$.\par
The aim of the present paper is to show that the possibility to write an integrable system in the Lax form, as well as the algebro-geometric technique related to this possibility, may also be applied to study qualitative features of the system, in particular its singularities.
}
\section*{Introduction}
It is well known that the majority of finite dimensional integrable systems can be written in the form
\begin{align}\label{Lax}
\diffXp{t}L(\lambda) = [L(\lambda), A(\lambda)]
\end{align}
where $L$ and $A$ are matrices depending on the time $t$ and  additional parameter $\lambda$. The parameter $\lambda$ is called a \textit{spectral parameter}, and equation \eqref{Lax} is called a \textit{Lax equation} with spectral parameter\footnote{More generally, any equation of the form $\dot L = [L,A]$ where $L,A$ are operators is called a Lax equation. Equations of such kind were first considered by Lax \cite{Lax} in connection with the KdV equation. Lax equations with spectral parameter were first considered in \cite{Novikov}.}.\par The possibility to write a system in the Lax form allows us to solve it explicitly by means of algebro-geometric technique. The algebro-geometric scheme of solving Lax equations can be briefly described as follows. Let us assume that the dependence on $\lambda$ is polynomial. Then, with each matrix polynomial $L$, there is an associated algebraic curve \begin{align}\label{SpectralCurve} C(L) = \{(\lambda,\mu) \in \Complex^2 \mid\det(L(\lambda)-\mu \E) = 0\} \end{align} called the \textit{spectral curve}. The Lax equation implies that this curve does not depend on time. Consider the set $\pazocal S_\Curve$ of matrix polynomials having the same spectral curve $ \Curve$. For each $L \in \pazocal S_\Curve$, there is an associated linear bundle over $\Curve$. This bundle is obtained by considering for each point $(\lambda,\mu) \in \Curve$ the kernel of the operator $L(\lambda)-\mu \E$, where $\E$ stands for the identity matrix. In this way, we obtain a map from $\pazocal S_\Curve$ to the Jacobian variety of the spectral curve. The classical result is that this map linearizes the Lax flow. For details, see e.g. the reviews \nolinebreak \cite{DMN, DKN, ReimanRev}, as well as references therein and Section \nolinebreak \ref{DSC} of the present paper. \par\smallskip
The aim of our paper is to show that the possibility to write an integrable system in the Lax form, as well as the algebro-geometric technique related to this possibility, may also be applied to study qualitative features of the system. \par In the last 30 years, there has been considerable interest in topology of singular Lagrangian fibrations associated to integrable systems \nolinebreak \cite{Kharlamov,Lerman, Lerman2, Delzant, FZ, AL1, AL2,PVN,leung2010almost, intsys, bolosh}. The generic structure of such fibrations is described by the classical Arnold-Liouville theorem which asserts that the phase space of an integrable system is almost everywhere foliated into invariant tori\footnote{In this paper, we mostly deal with integrable systems given by holomorphic functions on complex manifolds, so, formally speaking, there are no tori. However, the systems we deal with have a property of being \textit{algebraically completely integrable}, so that each of their regular invariant sets can be identified with an open subset in a certain Abelian variety, which is already a torus. }. This description breaks down on the singular set, that is the subset of the phase space where the first integrals become dependent.
%It is well known that the topology of a singular Lagrangian fibration associated to an integrable system is mainly determined by the \textit{singularities} of the system.
%
%
%, in particular, its singularities. By \textit{singularities} of an integrable system, we mean those points of the phase space where the first integrals become dependent (see \cite{intsys, bolosh} and references therein). 
Though the set of such points is of measure zero, these are singularities which mainly determine the global topology of the system. Furthermore, the most remarkable solutions, such as fixed points, stable periodic trajectories, and heteroclinic connections, belong to the singular set. Apart from this, singularities also arise in problems of quantization \nolinebreak \cite{MirHam, PVN2}, nearly-integrable systems \nolinebreak\cite{zung1996kolmogorov, Dullin}, and mirror symmetry \nolinebreak \cite{castano2010symmetries}. \par
Although singularities have been studied for a long time, their relation to algebro-geometric description of integrable systems seems to be not well understood.  We note that singularities of integrable systems can be, in principle, described by means of straightforward computations using explicit formulas for commuting Hamiltonians\footnote{We should mention that there also exists another approach to singularities of integrable systems based on the notion of a \textit{bi-Hamiltonian structure}. See \cite{Bolsinov,biham,SBSn}.}. However, firstly, these computations are rather tedious even for low-dimensional systems and, secondly, they do not allow us to see the relation between singularities and algebraic geometry related to the problem. Since first integrals of most of the known integrable systems arise as coefficients of an algebraic curve equation, and the solutions of these systems are expressed in terms of theta functions related to that curve, it seems to be inconsistent to ignore algebraic geometry when studying singularities\footnote{We note that topology of integrable systems, with no relation to singularities, was studied from the algebro-geometric point of view by Audin and her collaborators \cite{Audin, audin2}. See also our work \cite{LaxStab} where we discuss the relation between algebraic geometry and stability of solutions of integrable systems. }.
In this paper, we show that singularities naturally fit it the classical algebro-geometric scheme of solving Lax equations.\par \smallskip
Let us get down to the details. Let $m,n \in \mathbb{N}^*$ be positive integers, and let $J \in \gl(n,\Complex)$ be a fixed matrix with distinct eigenvalues. Consider the space \begin{equation*}\mathcal L_m^J (\gl(n,\Complex))= \left\{  \sum_{i=0}^{m} L_i\lambda^i \mid  L_i \in \gl(n,\Complex), \,L_m = J \right\} %\subset \gl(n,\Complex) \otimes \Complex[\lambda]
\end{equation*} 
of matrix-valued polynomials of degree $m$ with a fixed leading term $J$. It is well known that this space has a structure of a Poisson manifold. The Poisson structure on $\mathcal L_m^J (\gl(n,\Complex))$ is related to the decomposition of the loop algebra $\gl(n,\Complex) \otimes \Complex[\lambda, \lambda^{-1}] $ into a sum of two subalgebras  \cite{Reyman1,Reyman2}.
%, so that the space $\pazocal O(\mathcal L_m^J (\gl(n,\Complex)))$ of holomorphic functions on $\mathcal L_m^J (\gl(n,\Complex))$ has a structure of a Lie algebra.
% The Poisson structure on $\mathcal L_m^J (\gl(n,\Complex))$ has $mn$ independent Casimir functions, so that generic symplectic leaves are of dimension $mn(n-1)$. \par
 The Poisson bracket turns the space of holomorphic functions on $\mathcal L_m^J (\gl(n,\Complex))$ into a Lie algebra. This Lie algebra has a natural large commutative subalgebra.
Namely, let $\psi \in \Complex[\mu, \lambda^{-1}]$ be a polynomial in $\mu$ and $\lambda^{-1}$. Define a holomorphic function $H_\psi \colon \mathcal L_m^J (\gl(n,\Complex)) \to \Complex$ by the following formula:
\begin{align}\label{HamI}
H_{\psi}(L)= \Res_{\lambda = 0}\,\lambda^{-1}\,\Tr\psi(L(\lambda), \lambda^{-1}).
\end{align}
%where $1 \leq k \leq n$ and $0 \leq j < mk$, 
%where $\psi \in \Complex[\mu, \lambda^{-1}]$, 
Then for each $\psi_1, \psi_2 \in \Complex[\mu, \lambda^{-1}]$ we have $\{H_{\psi_1}, H_{\psi_2}\} = 0$, so that the space $$\pazocal F = \{H_\psi \mid \psi \in \Complex[\mu, \lambda^{-1}] \} $$ is a Poisson-commutative subalgebra of the space of holomorphic functions on $\mathcal L_m^J (\gl(n,\Complex))$. Moreover, $\pazocal F$ is an \textit{integrable system}, which means that the space $$\diff \pazocal F(L) = \{ \diff H(L) \mid H \in \pazocal F\}$$ is maximal isotropic at almost every point $L \in \mathcal L_m^J (\gl(n,\Complex)) $, see \cite{MF} for the $m=1$ case, in which this construction coincides with the so-called \textit{argument shift method}, and \cite{Reyman1,Reyman2, adler2} for the general case. \par 
In order to relate the above definition of integrability with the classical one, let us consider the functions
$$
H_{jk}(L)= %H_{\mu^k\lambda^{-j}}
\Res_{\lambda = 0}\,\lambda^{-1}\,\Tr \lambda^{-j}L(\lambda)^k
$$
where $1 \leq k \leq n$ and $0 \leq j < mk$. It is claimed that the functions $H_{jk}$ Poisson commute and are independent almost everywhere. Among the Hamiltonians $H_{jk}$, there are $mn$ Casimir functions, and the number of the remaining functions equals
$
\frac{1}{2}{mn(n-1)} $,
that is exactly one half of the dimension of a generic symplectic leaf. Therefore, Hamiltonian flows generated by each of the functions $H_{jk}$ are completely integrable in the Liouville sense\footnote{We note that some definitions of Liouville integrability include the requirement of completeness of Hamiltonian flows. In our case, this requirement is not satisfied.}. Note that the functions $H_{jk}$ are ``generators'' of the family $\F$, which means that each function $H_\psi \in \F$ is a function of $H_{jk}$'s.\par

For each $H_\psi \in \F$, the  Hamiltonian flow corresponding to $H_\psi$ have the Lax form
\begin{align}\label{loopI}
\diffXp{t}L(\lambda) = [L(\lambda), \phi(L(\lambda), \lambda^{-1})_+]
\end{align}
where $\phi = \diffFXi{\psi}{\mu}$ and $(\,\dots)_+$ denotes the sum of the terms of positive degree. \par As it is mentioned above, with each matrix polynomial $L \in \mathcal L_m^J$, we can associate the curve \nolinebreak \eqref{SpectralCurve} called the spectral curve. For each fixed curve $C$, the isospectral set
\begin{align}\label{isospec}
\pazocal S_C =  \{L \in \mathcal L_m^J (\gl(n,\Complex)) \mid  \Curve(L) = C \}
\end{align}
is preserved by each of the flows \eqref{loopI}. As it is easy to see, coefficients of the spectral curve equation are linear combinations of $H_{jk}$'s and vice versa, so $\pazocal S_C $ coincides with a common level set of Hamiltonians $H_{jk}$. The fibration
$
\mathcal L_m^J = \bigsqcup \pazocal S_C,
$
 where $C$ varies in the set of affine algebraic curves, is a \textit{singular Lagrangian fibration}. A fiber $ \pazocal S_C$ is called \textit{regular} if each point $L \in S_C$ is non-singular for the integrable system $\F$, i.e. if the Hamiltonians $H_{jk}$ are independent at each point $L \in \pazocal S_C $ (or, which is the same, the space $\diff \F(L)$ is maximal isotropic). Each regular fiber $\pazocal S_C$ is smooth, moreover it is a Lagrangian submanifold in the ambient symplectic leaf of $\mathcal L_m^J$. Fibers which are not regular are called \textit{singular}.
 \par
 As is well known, if the curve $C$ is non-singular, then the fiber $ \pazocal S_C$ is non-singular as well. Furthermore, in this case $ \pazocal S_C$ can be explicitly described as an open dense subset in the total space of a principal $(\Complex^*)^{n-1}$-bundle over the Jacobian of $C$, and the mapping $ \pazocal S_C \to \Jac(C)$ linearizes each of the flows \eqref{loopI}.\par
If the curve $C$ is singular, then some points $L \in \pazocal S_C$ may become singular, which means that the differentials of the Hamiltonians $H_{jk}$ become dependent.
The goal of this paper is to describe singularities arising on $ \pazocal S_C$ when the curve $C$ is nodal.\par\smallskip

The first part of the paper (Section 1) is devoted to the description of the set $\pazocal S_\Curve$ itself. Namely, we show that if $C$ is a nodal, possibly reducible curve, then $\pazocal S_\Curve$ is subdivided into natural smooth strata indexed by partial normalizations of $C$ and integer points in a certain convex polytope. For each stratum, there is a map to the generalized Jacobian of the corresponding partial normalization, and the image of each of the flows \eqref{loopI} under this map is a linear flow. Main result of this part of the paper is Theorem \ref{thm1} (see Section \ref{DSC}).
%This result is used to find the rank of each point $L \in \pazocal S_C$, i.e. the dimension of the space spanned by the differentials of $H_{jk}$'s at $L$.

 \par \smallskip In the second part of the paper (Section 2), we prove that if the spectral curve $C$ is nodal, then all singular points on $\pazocal S_\Curve$ are non-degenerate. Non-degenerate singularities of integrable systems are in a sense analogous to Morse singular points of smooth functions. In the complex case, all non-degenerate singularities of the same rank are locally symplectomorphic to each other. In the real case, each non-degenerate singularity can be represented as a product of three basic singularities: elliptic, hyperbolic, and focus-focus. In the same way, there are three kinds of nodal points of \textit{real} algebraic curves: acnodes (isolated points in the real part of the curve), crunodes (double points in the real part), and nodes which do not belong to the real part. 
In the case when the system under consideration is real\footnote{We note that the real version of the integrable system $\F$ discussed above is constructed in exactly the same way. Its phase space is $\mathcal L_m^J (\gl(n,\R))= \left\{  \sum_{i=0}^{m} L_i\lambda^i \mid  L_i \in \gl(n,\R), \,L_m = J \right\}. %\subset \gl(n,\Complex) \otimes \Complex[\lambda]
$
The Hamiltonians are of the same form \eqref{HamI} where the polynomial $\psi$ is real, and the corresponding Hamiltonian flows have the form \nolinebreak \eqref{loopI} where $\phi$ is also real. For each $L \in \mathcal L_m^J (\gl(n,\R))$, the associated spectral curve $C(L)$ is a \textit{real algebraic curve}, i.e. an affine algebraic curve over $\Complex$ endowed with an antiholomorphic involution $(\lambda, \mu) \to (\bar \lambda, \bar \mu)$.}, we show that acnodes, crunodes, and complex nodes in the spectral curve correspond to elliptic, hyperbolic, and focus-focus singularities respectively. Main results of this part are Theorems \ref{rkFormula} and \ref{NDT} (see Section \nolinebreak \ref{ncns}).

%We note that the real version of the integrable system $\F$ discussed above is constructed in exactly the same way. Its phase space is 
% \begin{equation*}\mathcal L_m^J (\gl(n,\R))= \left\{  \sum_{i=0}^{m} L_i\lambda^i \mid  L_i \in \gl(n,\R), \,L_m = J \right\}. %\subset \gl(n,\Complex) \otimes \Complex[\lambda]
%\end{equation*} 
%The Hamiltonians are of the same form \eqref{HamI} where the polynomial $\psi$ is real, and the corresponding Hamiltonian flows have the form \eqref{loopI} where $\phi$ is also real. For each $L \in \mathcal L_m^J (\gl(n,\R))$, the associated spectral curve $C(L)$ is a \textit{real algebraic curve}, i.e. an affine algebraic curve over $\Complex$ endowed with an antiholomorphic involution $(\lambda, \mu) \to (\bar \lambda, \bar \mu)$.\par
%We show that the number of elliptic, hyperbolic and focus-focus components 

%We show that acnodes, crusades 
%We show that elliptic components in this decomposition correspond to acnodes 

%The relation between integrable systems and algebraic geometry is known since 
\par\smallskip
%they do not allow us to see the underlying algebraic geometry of the problem.

 %We also give formulas for eigenvalues of linearized flows in terms of residues of meromorphic differentials on $C$.
Let us make one important remark.
The Hamiltonians \eqref{HamI} and equations \eqref{loopI} may seem to be of a rather special form. Nevertheless, it turns out that almost all known finite-dimensional integrable systems can be written in this form (see \cite{ReimanRev} and references therein), so the construction discussed is quite universal. However, in order to obtain physically interesting examples, we need to pass to a certain subspace $ \mathcal L' \subset  \mathcal L_m^J (\gl(n,\Complex))$. There are many natural subspaces $  \mathcal L'  \subset  \mathcal L_m^J (\gl(n,\Complex))$ which are Poisson manifolds, and if we pick those flows  \eqref{loopI}  which leave the subspace $  \mathcal L' $ invariant, we obtain an integrable hierarchy\footnote{For example, consider the subspace $\mathcal   L'  = \mathcal L_m^J (\so(n,\R)) \subset \mathcal L_m^J (\gl(n,\R)) $ which consists of those polynomials $L(\lambda)$ which satisfy
$
L(\lambda)^t = -L(\lambda).
$
As it is easy to see, the space $\mathcal L_m^J (\so(n,\R))$ is invariant with respect to the flow  \eqref{loopI}  if and only if the polynomial $\phi$ is real and odd in the variable $\mu$. Considering the flows  \nolinebreak\eqref{loopI}  for all such polynomials $\phi$, we obtain a completely integrable system on  
$\mathcal L_m^J (\so(n,\R))$. In particular, taking $m=2$ and $n=3$, we obtain the Lagrange top \cite{ratiu1982lagrange}.} on \nolinebreak $  \mathcal L' $.
In this paper, we only consider the hierarchy $\F$ either on the whole space $\mathcal L_m^J (\gl(n,\Complex))$, or on its real counterpart $\mathcal L_m^J (\gl(n,\R))$. However, almost all of our results, in particular Theorems \ref{rkFormula} and \ref{NDT}, can be extended to restricted systems in a more or less straightforward way.
%\pagebreak[4]
%\displaybreak[1]
 In particular, we claim that  singularities of such classical integrable systems as Euler, Lagrange and Kovalevskaya tops, spherical pendulum, geodesic flow on ellipsoid etc., as well as their multidimensional generalizations, can be described using our approach.

\section{Singular spectral curves, generalized Jacobians, and convex polytopes }\label{sec1}
\subsection{Description of the set $\pazocal S_C$}\label{DSC}
In this section, we assume that  $C$ is a nodal curve and describe the set $\pazocal S_C$ defined by \nolinebreak \eqref{isospec}. {We note that the simplicity of the spectrum of the leading term $J$ implies that the spectral curve $C$ is necessarily reduced, i.e. its defining polynomial has no multiple factors.}\par
It is clear that the curve $C$ should satisfy some additional assumptions in order for the set $\pazocal S_C$  to be non-empty. Namely, let $\mathcal C_{spec} $ be the set of plane affine algebraic curves in with defining polynomial $\chi(\lambda, \mu)$ satisfying
\begin{equation*}
\lim_{z \to 0}\left(\chi \left(\frac{1}{z}, \frac{w}{z^m}\right)z^{nm} \right) =
%\left. \det(L_\gamma + zL_{\gamma-1} + \dots + z^\gamma L_0 - w\E)\right\vert_{z=0} =
 \det(J - w\E).
 \end{equation*}

%Their common level sets
%$$\pazocal S(C) = \{L(\lambda) \in \mathcal L_m^J (\gl(n,\Complex))  : \{\det(L(\lambda)-\mu \E) = 0\} = C\}$$

%because restrictions of  \eqref{loop} to various natural subspaces of $\mathcal L_m^J (\gl(n,\Complex))$ gives rise to almost all known spinning top-like integrable systems.

%are Hamiltonian and constitute an integrable hierarchy (see  \cite{Reyman1,Reyman2}). 

%are invariant  under each of the flows \eqref{loop}. 
Clearly, for $\pazocal S_C$ to be non-empty, we should have $C \in  \mathcal C_{spec}$.
%\begin{equation}\label{PCond}\left.z^{nm}p\left(\frac{1}{z}, \frac{w}{z^m}\right)\right\rvert_{z=0}  =
%%\left. \det(L_\gamma + zL_{\gamma-1} + \dots + z^\gamma L_0 - w\E)\right\vert_{z=0} =
% \det(J - w\E).\end{equation}
So, in what follows, we consider the set  $\pazocal S_C$ only for $C \in  \mathcal C_{spec}$.\par
First, assume that the spectral curve $C$ is non-singular. The description of the set $\pazocal S_C$ in this case is well-known. Namely, consider the Riemann surface $X$ which is obtained from the spectral curve $C$ by adding points at infinity. Let $$\PGL(\Complex, J) =\{ R \in \PGL(n,\Complex) \mid RJ = JR\}.$$ The set $\pazocal S_C$ carries the natural action of $\PGL(\Complex,J)$ by conjugation. This action is free and preserves each of the flows \eqref{loopI}. Let $\hat{\pazocal S}_C = \quo{\pazocal  S_C}{\PGL(\Complex, J)}.$ Then, as shown in \cite{mvm, adler, Reiman2}, there exists a biholomorphic map
$$
\hat \Phi \colon \hat{\pazocal S}_C \to \Pic_{g + n -1}(X) \setminus (\Theta_{g-1} + [D_\infty])
$$
where $g$ is the genus of $X$, $\Theta_{g-1} \subset \Pic_{g  -1}(X)$ is the theta divisor, and $D_\infty$ is the pole divisor of $\lambda$.  Furthermore, the image of the flow \eqref{loopI} under the map $\hat \Phi$ is a linear flow given by
%Van Moerbeke and Mumford \cite{mvm} showed that there exists a biholomorphic mapping $\hat \Phi$ from  $ {S}^r(p) $ to the open dense subset of  $\Pic_{g'}(X(p))$ where $g' = g+n-1$, and $g$ is the genus of $X(p)$. 
% Furthermore, as shown in Adler and van Moerbeke \cite{adler2, adler}, the image of the flow \eqref{loop} under the map $\hat \Phi$ is a linear flow given by
\begin{align}\label{velocity}
\omega\left(\diffFXYp{\xi}{t}\right) = \sum_{P: \,\lambda(P) = \infty} \Res_{P}\, \phi\omega
\end{align}
where $ \xi \in \Pic_{g+n-1}(X)$, $\omega \in \Omega^1(X)$, and the cotangent space to  $\Pic_{g+n-1}(X)$ is identified with the space $ \Omega^1(X)
$ of holomorphic differentials on $X$.
%and $\infty_1, \dots, \infty_n$ are the poles of $\lambda$.
\par
The set $\pazocal S_C$ itself has a structure of a holomorphic principal $\PGL(\Complex,J)$-bundle over $\hat{\pazocal S}_C$. The structure of this bundle is described in \cite{Gavrilov2}. Let $\infty_1, \dots, \infty_n$ be the poles of $\lambda$, and let $ X'$ be the curve obtained from $X$ by identifying $\infty_1 \sim \dots \sim \infty_n$. Then there exists a biholomorphic map 
$$
\Phi \colon {\pazocal S}_C \to \Pic_{g + n -1}(X') \setminus \pi^{-1}(\Theta_{g-1} + [D_\infty])
$$
where $\pi$ is the natural projection $ \Pic_{g+n-1}(X') \to  \Pic_{g+n-1}(X)$. The projection $\pi$ defines a principal bundle structure on $\Pic_{g+n-1}(X')$, and the following diagram commutes:\begin{align*}
       				 \begin{CD}
           			 	\pazocal S_C@> \Phi  >> \Pic_{g+n-1}(X')  \\
            				@VV   V  @VV  \pi V \\
          				 \hat{\pazocal S}_C  @> \hat \Phi>> \Pic_{g+n-1}(X)
       			 \end{CD}
  \end{align*}
 The image of the flow \eqref{loopI} under the map $\Phi$ is given by the same formula \eqref{velocity} where $\omega$ is now not necessarily holomorphic, but may have poles of the first order at points $\infty_1, \dots , \infty_n$.\par
The goal of this part of the paper is to extend these results to the case when the spectral curve is nodal and possibly reducible. Recall that a singular point $(\lambda,\mu)$ of a plane affine algebraic curve $\{ \lambda, \mu \in \Complex^2 \mid p(\lambda,\mu) = 0\}$ is called a node, or an ordinary double point, if the Hessian $\diff^2 p(\lambda,\mu)$ is non-degenerate. A plane algebraic curve is called nodal if all its singular points are nodes. See Section \ref{nc} for details on nodal curves.
\par
Let $C$ be a nodal, possibly reducible curve, and let $\mathrm{Sing}\, C$ be the set of its nodes.  Let \nolinebreak $L \in \pazocal S_C$. As is well known (see \cite{babelon}, Chapter 5.2), for all points $(\lambda, \mu) \in C \setminus \Sing C$, we have $$\dim \Ker(L(\lambda)-\mu \E) = 1.$$ For any $(\lambda, \mu) \in \Sing C$, there are two possibilities:  
 $$ \mbox{}\quad \dim \Ker(L(\lambda)-\mu \E) = 1 \quad \mbox{or}\quad\dim\Ker(L(\lambda)-\mu \E) = 2.$$
This dichotomy gives rise to a natural stratification of $\pazocal S_C$. For $L \in \pazocal S_C$, let
$$
K(L) = \{ (\lambda, \mu) \in \Sing C(L) \mid \dim \Ker(L(\lambda)-\mu \E) = 2\}.
$$
 Let $K \subset \Sing C$, and let
$$\pazocal S_{C}^K = \{L \in \pazocal S_C \mid K(L) = K \}.$$
The set $\pazocal S_{C}^K$ is a quasi-affine variety. 
%Let also 
%$$
%\pazocal S_{C}^{reg} =  \{L \in \pazocal S_C \mid K(L) = \emptyset\}.
%$$
Clearly, we have
\begin{align}\label{scstrat}
\pazocal S_{C}=\bigsqcup_{K \subset \Sing C}\pazocal S_{C}^K
\end{align}
where the union is disjoint in set-theoretical, not in topological sense.
Stratification \eqref{scstrat} is preserved by each of the flows \eqref{loopI}, that is all these flows leave $\pazocal S_{C}^K$ invariant for each $K$. Below, we give a geometric characterization of the sets $\pazocal S_{C}^K$.\par
%We also define $\pazocal S_{C}^K^r$ as the quotient of $\pazocal S_{C}^K$ by the $\PGL(\Complex,J)$ action. We note that this action is no longer free, however there always exists a subgroup $H \subset \PGL(\Complex,J)$ such that $\PGL(\Complex,J) / H$ acts freely, and $H$ acts trivially.\par
Let $X_K$ be the curve which is obtained from $C$ by adding points at infinity and blowing up at the points of $K$, and let $X'_K$ be the curve obtained from $X_{K}$ by identifying $\infty_1 \sim \dots \sim \infty_n$. Our first result is that  $\pazocal S_{C}^K$ is biholomorphic to an open subset in the disjoint union of $r$ copies of the generalized Jacobian of $X'_{K}$ where $r$ is the number of integer points in a certain polytope constructed from the curve  $X'_{K}$.  Let us describe the construction of this polytope.\par
 Let $Y$ be a curve, and let $Y_1, \dots, Y_k$ be its irreducible components. A \textit{multidegree} on $Y$ is a mapping $d \colon \{ Y_1, \dots, Y_k\} \to \Z$. The total degree of a multidegree $d$ is the number $|d| = \sum_{i=1}^k d(Y_i).$ For each $I \subset \{ 1,\dots, k\}$, define the \textit{subcurve}  $Y_I \subset Y$ as the union $\bigcup_{i \in I}Y_i$. For each subcurve $Y_I \subset Y $ we can restrict $d$ on $Y_I$ and get a multidegree $d_{I}$ on $Y_I$ (see also Section \ref{pnsc}). \par
%Our next result concerns the unreduced set $\pazocal S_{C}^K$.\par
Let $d$ a multidegree on $Y$ of total degree $g(Y)$ where $g(Y)$ is the arithmetic genus of $Y$ (see Section \ref{nc}). We say that $d$ is \textit{uniform} if for each subcurve $Y_I \subset Y$ we have
$$
\lvert d_{I} \rvert \geq g(Y_I).
$$
% We will denote the set of all uniform multidegrees on $Y$ by $\Sigma^u(Y)$. 
The set of uniform multidegrees is non-empty and coincides with the set of integer points in the convex polytope
\begin{align}\label{polytope}
P = \left\{ x \in \R^k : \sum_{i=1}^k  x_i = g(Y); \,\, \sum_{i \in I} x_i \geq g\left(Y_I\right) \,\, \forall \,\, I \subset \{ 1, \dots, k\} \right\}.
\end{align}
Let $\Delta_K$ be the set of uniform multidegrees on $X'_K$. %As it is easy to see, this set is always non-empty.
\begin{theorem}\label{thm1}
Assume that $C \in \mathcal C_{spec}$ is a nodal curve, and let $K \subset \Sing C$. Then:
\begin{enumerate}
\item
The set $\pazocal S_{C}^K$ is a complex analytic manifold of dimension
$$
\dim \pazocal S_{C}^K = g(X'_K) = \frac{mn(n-1)}{2}  - |K|.
$$
\item There exists a biholomorphic map $$\Phi \colon  \pazocal S_{C}^K \to  \bigsqcup_{d  \in \Delta_K }\left(\Pic_{d}(X'_{K}) \setminus \Upsilon_d \right) $$
where $ \Delta_K$ is the set of uniform multidegrees on $X'_K$, and $\Upsilon_d \subset \Pic_{d}(X'_{K})$ is a subset of positive codimension.
%, i.e. all lattice points in the polytope $\pazocal P(X'_{p,K})$.
%with an open dense image.
%where the union is taken over all semistable multidegrees of total degree $g(X'_{p,K})$.%, and $S$ is closed of positive codimension;
\item
The image of the flow \eqref{loopI} under the mapping $\Phi$ is given by the formula \eqref{velocity} where $\omega$ is any regular differential on $X'_{K}$.
\item  The flows \eqref{loopI} span the tangent space to $\pazocal S_{C}^K$ at every point.
\end{enumerate}
\end{theorem}
%As for the set $\pazocal S_C$ itself, we have the following.
%Let us also consider the set
%$$
%\pazocal S_{C}^r = \bigsqcup_{|K| \geq r} \pazocal S_{C}^K
%$$
\begin{remark} We note that Theorem \ref{thm1} remains true over the reals. Namely, the set $$\Ree \pazocal S_{C}^K =  \pazocal S_{C}^K \cap \mathcal L_m^J (\gl(n,\R))$$ is a smooth manifold diffeomorphic to an open subset in the real part of  $\bigsqcup_{d  \in \Delta_k }\Pic_{d}(X'_{K}) $.
%Note that the real part of $\Pic_{d}(X'_{K}) $ is never empty, but may be disconnected, so $\Ree \pazocal S_{C}^K$ may have more connected components than $\pazocal S_{C}^K $.
%\end{remark}
%\begin{remark}
%
\end{remark}
\begin{remark}
Regular differentials on $X'_K$ can be viewed as meromorphic differentials on its non-singular compact model $X$ with special properties. Namely, a meromorphic differential $\omega$ on $X$ is regular on  $X'_K$ if all poles of $\omega$ are simple, and for each $Q \in X'_K$, we have
%explicitly described as follows. Let $\omega$ be a differential on the smooth part of $X'_K$, and let $\pi \colon X \to X'_K$ be the normalization map. The differential $\omega$ is regular if its lift $\pi^*\omega$ can be extended to a meromorphic differential on $X$ which has the following property: 
$$
\sum_{P: \pi(P) = Q} \Res_P \,\omega = 0
$$
where $\pi \colon X \to X'_K$ is the normalization map.
See Section \ref{nc} for details on regular differentials on curves. 
\end{remark}
\begin{consequence}\label{corIrCom}
Assume that $C \in \mathcal C_{spec}$ is a nodal curve. Then:
\begin{enumerate}
\item  The dimension of $\pazocal S_{C}$ is equal to $\frac{1}{2}mn(n-1)$.
\item  The number of irreducible components of $\pazocal S_{C}$ is equal to the number of uniform multidegrees on the curve $X'_\emptyset$.
%\item Let $K \subset \Sing C$. Then 
% Let $K_1,K_2 \subset \Sing C$. Then $\pazocal S_{C}^{K_1}$ lies in the closure of $\pazocal S_{C}^{K_2}$ if and only if $K_1 \supset K_2$. In other words, the set $\overline{\pazocal S_{C}^K}$ is the closure of $\pazocal S_{C}^K$.
%\item  The number of irreducible components of $\overline{\pazocal S_{C}^K}$ is equal to the number of uniform multidegrees on the curve $X'_K$.
\end{enumerate}
\end{consequence}
\nopagebreak
\begin{proof}
Assertion 1 is obvious. Assertion 2 is proved in Section \ref{ncns}.
\end{proof}
The proof of Theorem \ref{thm1} is given in Section \ref{proof1}. 
In Section \ref{ex1}, we consider an example. Namely, we take $m=1$ and discuss the set $\pazocal S_C^\emptyset$ when the spectral curve $C$ is $n$ straight lines in general position. In this case, the polytope \eqref{polytope} turns out to be the \textit{permutohedron} $P_n$. It turns out that solutions of \eqref{loopI} corresponding to $n!$ vertices of $P_n$ lie in Borel subalgebras containing the centralizer of $J$. These solutions are linear combinations of exponents, which means that if $d$ is a vertex of $P_n$, then $\Pic_{d}(X'_{\emptyset})$  is completely contained in the image of \nolinebreak $\Phi$, i.e. the exceptional set $\Upsilon_d$ is empty. For integer points in the interior of $P_n$, this is no longer so, and the solutions turn out to be rational functions of exponents.\par\smallskip
%Let us also define 
%$$
%\overline{\pazocal S_{C}^K}=\bigsqcup_{K' \supseteq K}\pazocal S_{C}^K.
%$$

%In particular, we always have $\dim \pazocal S_C$
%

%(there are exactly $n!$ of them). %while solutions corresponding to integer points in the interior of the permutohedron are of different nature.

 %In the case $n=3$, we also write down solutions corresponding to 

Let us also define $\hat{\pazocal S}_{C}^K$ as the quotient of $\pazocal S_{C}^K$ by the $\PGL(\Complex,J)$ action. We note that this action is no longer free, however there always exists a subgroup $H \subset \PGL(\Complex,J)$ such that $\PGL(\Complex,J) / H$ acts freely, and $H$ acts trivially.%Our next result is the description of the reduced set $\pazocal S_{C}^K^r$.
\par
The following statement follows from Theorem \ref{thm1}.
\begin{consequence}
Assume that $C \in \mathcal C_{spec}$ is a nodal curve, and that $K \subset \Sing C$. Then:
\begin{enumerate}
\item
The set $\hat{\pazocal S}_{C}^{K}$ is a complex analytic manifold of dimension
$$
\dim \hat{\pazocal S}_{C}^{K} = g(X_K) = \frac{mn(n-1)}{2} - n - |K| + \dim \Hom_0(X_K).
$$
\item There exists a biholomorphic map 
$$\hat \Phi \colon  \hat{\pazocal S}_{C}^K \to  \bigsqcup_{d  \in \Delta_K }\left(\Pic_{d}(X_{K}) \setminus \hat \Upsilon_d \right) $$
where $ \Delta_K$ is the set of uniform multidegrees on $X'_K$ (not $X_K$!), and $\hat \Upsilon_d \subset \Pic_{d}(X_{K})$ is a subset of positive codimension

%, i.e. all lattice points in the polytope $\pazocal P(X'_{p,K})$.
%with an open dense image.
%where the union is taken over all semistable multidegrees of total degree $g(X'_{p,K})$.%, and $S$ is closed of positive codimension;
\item
The image of the flow \eqref{loopI} under the mapping $\hat \Phi$ is given by the formula \eqref{velocity} where $\omega$ is any regular differential on $X_{K}$.
\item  The flows \eqref{loopI} span the tangent space to $\hat{\pazocal S}_{C}^{K}$ at every point.
\item Let $\pi \colon \Pic(X'_{K}) \to \Pic(X_{K}) $ be the natural projection. Then the following diagram commutes:
\begin{align*}
       				 \begin{CD}
           			 	\pazocal S_{C}^K@> \Phi  >> \Pic(X'_K)  \\
            				@VV   V  @VV  \pi V \\
          				 \hat{\pazocal S}_{C}^{K}  @> \hat \Phi>> \Pic(X_K)
       			 \end{CD}
  \end{align*}
%  In particular, the sets $\Delta_K$ and $\Delta'_K$ coincide.
\end{enumerate}
\end{consequence}
%We do not prove Theorem 2, since it is a simple corollary of Theorem 1.
\begin{remark}
We note that it is also possible to describe the set $\Delta_K$ in terms of the curve $X_K$ itself.\par
Let $Y$ be a nodal curve, and let $d$ a multidegree on $Y$ such that $|d| = g(Y) - \dim \Hom_0(Y)$. Then $d$ is called \textit{semistable} if for each subcurve $Y_I \subset Y$ we have
\begin{align}\label{semistab}
\lvert d_{I} \rvert \geq g(Y_I) -  \dim \Hom_0(Y_I).
\end{align}
A condition equivalent to \eqref{semistab} first appeared in \cite{beauville}. The term \textit{semistable multidegree} is suggested in \cite{Alexeev}.\\ %The set $ \Sigma^{ss}(Y)$ coincides with the set of integer points in the convex polytope
%$$
% \Sigma^{ss}_{\R}(Y)=\left\{ x \in \R^k : \sum_{i=1}^k  x_i = g(Y) -  \dim \Hom_0(Y); \,\,  \sum_{i \in I} x_i \geq g\left(Y_I\right)  -  \dim \Hom_0(Y_I)\,\, \forall \,\, I \subset \{ 1, \dots, k\} \right\}.
%$$
%Let us consider $ \Sigma^{ss}(X_K)$, and let
As it is to see (see Proposition \ref{uss}), we have
$$\Delta_K = \{ d \in \Z^k \mbox{ such that } d - \deg D_\infty \mbox{ is semistable} \}.$$
%The set $\Delta_K$ coincides with the set of integer points in a shifted polytope $ \Sigma^{ss}_\R(S_K) + \deg D_\infty$. 
We note that the multidegree $\deg D_\infty$ has a transparent description. If the defining polynomial of $C$ is $\chi = \chi_1\cdot \ldots \cdot \chi_k$, then $\deg D_\infty = (\deg_{\!\mu}\, \chi_1, \dots, \deg_{\!\mu}\, \chi_k)$.\par
%We note that the description of $\pazocal S^r_{C,K}$ in terms of semistable multidegrees suggests that there should exist a map from  $\pazocal S_{C}^K/\PGL(\Complex,J)$ to the \textit{canonical compactified Jacobian}
\end{remark}

\subsection{Nodal curves and generalized Jacobians}\label{nc}
The theory of generalized Jacobians is due to Rosenlicht \cite{Rosenlicht1, Rosenlicht}, see also Serre \cite{Serre}. %and Mumford \cite{Mumford}.
 In this section, we present an elementary exposition of this theory for nodal, possibly reducible, curves. We note that all the presented results are well-known, at least in the irreducible case. As for the reducible case, we were not able to find some statements in the literature, in particular, Proposition \ref{density} concerning effective Weil divisors on reducible curves.
\subsubsection{Nodal curves and arithmetic genus}
\begin{definition}
A plane affine algebraic curve $\{ \lambda, \mu \in \Complex^2 \mid p(z,w) = 0\}$ is called a \textit{plane nodal curve} if $\det \diff^2 p \neq 0$ at all points where $\diff p = 0$.
\end{definition}
 Below, we give a more abstract definition of nodal curves.\par
Let $X = X_1 \sqcup \ldots  \sqcup X_r$  where $r\geq 1$ be a disjoint union of connected Riemann surfaces, 
%We shall use the notation $g(X_i)$ for the genus of $X_i$, and $g(X) = \sum_{i=1}^k g(X_i)$ for the genus of $X$.  
and let $ \Sigma = \{ \pazocal P_1, \dots, \pazocal P_\sigma \} $ be a finite set of pairwise disjoint $2$-element subsets of $X$. Assume that $\pazocal P_i = \{ P_i^+, P_i^-\} $, and consider the topological space $X / \Sigma$ obtained from $X$ by identifying $ P_i^+$ with $P_i^-$  for each $1 \leq i \leq \sigma$.  
%We shall use the notation $n(X,\Sigma) = \sigma$
%\begin{align*}
%g(X) \mbox{ is the genus of } X,\\
%n(X,\Sigma) = \sigma \mbox{ is the number of nodes of } X/\Sigma
%\end{align*}
\par
Let $\pi \colon X \to X / \Sigma$ be the natural projection, and let $\mathrm{supp}( \Sigma) = \bigcup_{i=1}^{\sigma}\pazocal P_i$. A function $f \colon X / \Sigma \to \CPP$ is called meromorphic on $X / \Sigma $ if $\pi^*(f)$ is a meromorphic function on $X$ which does not have poles at the points of $\mathrm{supp}(\Sigma)$.\par
\begin{definition} The space $X / \Sigma$ endowed with the described ring of meromorphic functions\footnote{Formally speaking, to turn $X / \Sigma$ into a complex analytic space, we should have described its structure sheaf. However, we do not need it.} is called a \textit{nodal curve}.
\end{definition} Obviously, a plane nodal curve completed at infinity is a nodal curve. The converse is of course not true, i.e. not any nodal curve can be obtained in this way.\par
%Of course, plane algebraic curves defined in Section \ref{DSC} 
%The above definition of nodal curves and meromorphic functions on them is related to the one given in Section \ref{DSC} in the following way. If $C = \{ \lambda, \mu \in \Complex^2 \mid p(\lambda,\mu) = 0\}$ is a plane nodal curve, then we take the the non-singular compact model of $C$ as $X$. Then, obviously, there exists $\Sigma$ such that 
\par
 In what follows, we prefer to work with the non-singular curve $X$ endowed with the set \nolinebreak$\Sigma$ rather than with the singular curve $X / \Sigma$. The terminology which we introduce below may seem to be non-standard, however it is quite convenient in the situation when we need to work with different singularizations of the same Riemann surface simultaneously. We also note that our approach to singular curves is rather close to the original approach of Rosenlicht.
 \begin{definition}\label{srf}
 We say that a function {on $X$} is $\Sigma$\textit{-regular} if it is meromorphic and takes same finite values at $P_i^\pm$. 
  \end{definition}
 Obviously, $\Sigma$-regular functions on $X$ are in one-to-one correspondence with meromorphic functions on $X / \Sigma$. The ring of $\Sigma$-regular functions will be denoted by $\pazocal M(X, \Sigma)$:%In what follows, we identify $f$ and $\pi^*(f)$. In other words, by a function meromorphic on $C / \Sigma$, we mean a function which is meromorphic on $C$ and takes same finite values at $P_i^\pm$. We will denote the ring of such functions by $\pazocal M(C, \Sigma)$:
%The ring of meromorphic functions $\pazocal M(C/\Sigma)$ can be identified with a subring of $\pazocal M(C)$ which we denote by $\pazocal M(C, \Sigma)$:
$$
\pazocal M(X,\Sigma) = \{ f \in \pazocal M(X) : f(P_i^+)=f(P_i^-) \neq \infty \,\,\forall\,\, 1 \leq i \leq |\Sigma|)\},
$$
where  $\pazocal M(X)$ is the ring of functions meromorphic on $X$. \par
%and instead of considering the singular curve $C / \Sigma$, we may consider its normalization $C$ equipped with a subring $\pazocal M(C, \Sigma) \subset \pazocal M(C)$. \par
 \begin{definition}\label{srd}
A meromorphic differential $\omega$ on $X$ is $\Sigma$\textit{-regular} if all its poles are simple, contained in $\mathrm{supp}(\Sigma)$, and
\begin{align}\label{resSum}
\mathrm {res}_{P_i^+} \,\omega + \mathrm {res}_{P_i^-} \,\omega = 0 \quad \forall\,\,  1 \leq i \leq |\Sigma|.\end{align}
  \end{definition}
   We denote the space of $\Sigma$-regular differentials by $\Omega^1(X,\Sigma)$. 
   %The space $\Omega^1(X,\Sigma)$ contains the subspace $\Omega^1(X)$ of differentials holomorphic on $X$.
 \begin{definition}The number \nolinebreak$\dim \Omega^1(X,\Sigma)$ is called the \textit{arithmetic genus} of $X / \Sigma$.\end{definition}
%Let also $ \Omega^1(X)$ be the space of holomorphic differentials on $X$. 
%\par \smallskip\smallskip

%We denote it by $g(X,\Sigma) $. 
Let us denote the arithmetic genus by $g(X,\Sigma)$. We shall also use the notation $g(X)$ for the geometric genus of $X$ (that is the dimension of the space $\Omega^1(X)$ of holomorphic differentials on $X$), $c(X)$ for the number of connected components of $X$, or, which is the same, number of irreducible components of $X/\Sigma$, and $c(X,\Sigma)$ for the number of connected components of $X/\Sigma$. These notations are summarized in Table \ref{table1}. \par 
\begin{table}
{
\centerline{\begin{tabular}{c|c}
Notation & Meaning \\
\Xhline{2\arrayrulewidth}
 $g(X)$ &  $\vphantom{\displaystyle \int^a_b}\dim \Omega^1(X) = \frac{1}{2}\dim \Hom_1(X,\Complex)$, genus of $X$ \\
\hline
% $\sigma(X,\Sigma)$ & $\vphantom{\displaystyle \int^a_b}\left\lvert\Sigma\right\rvert$ & number of nodes in $X/\Sigma$ \\
% \hline
 $c(X)\vphantom{\displaystyle \int^a_b}$&$\dim \Hom_0(X,\Complex)$, number of irreducible components of $X/\Sigma$\\
 \hline
 $g(X,\Sigma)$ & $\vphantom{\displaystyle \int^a_b}\dim \Omega^1(X,\Sigma)$, arithmetic genus of $X/\Sigma$\\ 
 \hline
 $c(X,\Sigma)$ & $\vphantom{\displaystyle \int^a_b}\dim \Hom_0(X/\Sigma,\Complex)$, number of connected components of $X/\Sigma$\\ 
%  \hline
% $\vphantom{\displaystyle \int^a_b} \pazocal M(X,\Sigma)$ & $\Sigma$-regular functions (see Definition \ref{srf})\\ 
%   \hline
% $\vphantom{\displaystyle \int^a_b} \Omega^1(X,\Sigma)$ & $\Sigma$-regular differentials (see Definition \ref{srd})\\ 
 \end{tabular}}
 }
 \caption{Notations}\label{table1}
 %\par\smallskip\smallskip
\end{table}
To count the arithmetic genus, we need the notion of the dual graph of a nodal curve. This graph has $r = c(X)$ vertices and $\sigma = |\Sigma|$ edges. 
%The vertices are in one-to-one correspondence with irreducible components 
The vertices $v_1, \dots, v_r$ of this graph correspond to irreducible components $X_1, \dots, X_r$, and the edges $e_1, \dots, e_\sigma$ correspond to nodes $P_1^\pm, \dots, P_\sigma^\pm$: if $P_i^- \in X_j$, and $P_i^+ \in X_k$, then there is an oriented\footnote{Of course, if we rename $P_i^-$ to $P_i^+$ and vice versa, the orientation will be reversed. However, it is convenient to assume that the orientation is fixed.} edge going from $v_j$ to $v_k$. We denote the dual graph of $X/\Sigma$ by $\Gamma(X,\Sigma)$ or, when it does not cause confusion, just $\Gamma$. We note that $$\dim \Hom_0(\Gamma, \Complex) = c(X,\Sigma) = \dim \Hom_0(X/\Sigma, \Complex). $$\par
With each $\omega \in \Omega^1(X,\Sigma)$, we associate a $1$-chain in the dual graph:
\begin{align*}
 Z(\omega) = \sum\nolimits_{i=1}^{|\Sigma|} \left(\Res_{P_i^+} \,\omega \right)e_i.
\end{align*}
\begin{statement}
For each $\omega \in \Omega^1(X,\Sigma)$, the chain $Z(\omega)$ is a cycle.
\end{statement}
\begin{proof}
This follows from condition \eqref{resSum}.
\end{proof}
The following is simple.
\begin{statement}\label{regDiff} \quad\par
\begin{enumerate}\item The mapping $Z \colon \quo{\Omega^1(X,\Sigma)}{\Omega^1(X)}\to \Hom_1(\Gamma(X,\Sigma), \Complex)$ is an isomorphism.
\item The arithmetic genus of a nodal curve is given by $$ g(X,\Sigma) = g(X) + \dim \Hom_1(\Gamma,\Complex) = g(X)+|\Sigma| + c(X,\Sigma) - c(X) .$$
\end{enumerate}
\end{statement}

\subsubsection{Generalized Jacobian of a nodal curve}
%\begin{statement}
%A $\Sigma$-regular divisor of multi-degree $0$ is sigma-principal if and only if 
%\end{statement}
%Two $\Sigma$-regular divisors are $\Sigma$-linearly equivalent if their difference is a $\Sigma$-principal divisor. We write: $$D_1 \stackrel{\Sigma}{\sim} D_2.$$
%The group of $\Sigma$-linear equivalence classes is the Picard variety $$

Let us define the generalized Jacobian of $X/\Sigma$. As in the non-singular case, there is a natural mapping
$$
 \mathcal I \colon \Hom_1(X\setminus \mathrm{supp}(\Sigma), \Z) \to \Omega^1(X,\Sigma)^*
$$
given by $$\langle \mathcal I(\gamma), \omega \rangle = \oint_\gamma \omega.$$ The image of this mapping is a lattice $L(X, \Sigma) \subset \Omega^1(X,\Sigma)^*$ called the \textit{period lattice}.
\begin{definition}
The quotient $$\Jac(X, \Sigma) = \quo{\Omega^1(X,\Sigma)^*}{L(X, \Sigma)}$$ is called the \textit{generalized Jacobian} of $X / \Sigma$. 
\end{definition}
Let us describe the structure of the generalized Jacobian.
By Proposition \ref{regDiff}, we have an exact sequence
\begin{align*}
\begin{CD}
           			 	0 @>>> \Hom^1(\Gamma,\Complex) @> Z^* >> \Omega^1(X,\Sigma)^* @> \pi^* >> \Omega^1(X)^* @>>> 0 \\
       			 \end{CD}
   			 \end{align*}
			 where $\pi^*$ is the restriction map. Obviously, $\pi^*$ maps $L(X,\Sigma)$ onto $L(X)$ where $L(X)$ is the usual period lattice of $X$. The kernel of the mapping $\pi^* \colon L(X,\Sigma) \to L(X)$ consists of integrals over cycles contractible in $X$, i.e. functionals $\xi$ of the form
			 $$
			\langle  \xi,\omega \rangle = 2\pi\mathrm{i}\sum\nolimits_{i=1}^{|\Sigma|} k_i\,\Res_{P_i^+} \omega, \quad k_i \in \Z.
			 $$ We get another exact sequence 
			 \begin{align*}
       				 \begin{CD}
				        0 @>>> \Hom^1(\Gamma, 2\pi  \mathrm{i}\Z) @> Z^* >> L(X,\Sigma) @> \pi^* >> L(X) @>>> 0. \\
       			 \end{CD}
   			 \end{align*}
			 Combining these two exact sequences, we get the following commutative diagram:
\begin{align}\label{9diag}
       				 \begin{CD}
				        @. 0 @. 0 @. 0 @. \\
				        @. @VVV @ VVV @ VVV @. \\
				        0 @>>> \Hom^1(\Gamma, 2\pi  \mathrm{i} \Z) @>Z^*>> L(X,\Sigma) @>\pi^*>> L(X) @>>> 0 \\
				        @. @VVV @ VVV @ VVV @. \\
           			 	0 @>>> \Hom^1(\Gamma,\Complex) @>Z^*>> \Omega^1(X,\Sigma)^* @>\pi^*>> \Omega^1(X)^* @>>> 0 \\
				  @. @VVV @ VVV @ VVV @. \\
				  0 @>>> \Hom^1(\Gamma,\Complex / 2\pi  \mathrm{i} \Z) @>Z^*>>\Jac(X,\Sigma) @>\pi^*>>\Jac(X) @>>> 0 \\
				  		        @. @VVV @ VVV @ VVV @. \\
						           @. 0 @. 0 @. 0 @. \\
            				%@VV \pazocal A_\Sigma V  @VV \pazocal A V \\
          				%\Jac(X, \Sigma) @> \pi^J>>\Jac(X)
       			 \end{CD}
   			 \end{align}
The columns and the two top rows of this diagram are exact, so the bottom row is exact as well.
We conclude that the generalized Jacobian $\Jac(X,\Sigma)$ is the extension of the usual Jacobian $\Jac(X)$ by the group $ \Hom^1(\Gamma,\Complex / 2\pi  \mathrm{i}\Z) \simeq (\Complex^*)^{m}$ where $$m = \dim \Hom_1(\Gamma,\Complex)= |\Sigma| - c(X) + c(X,\Sigma).$$
\subsubsection{Abel map}
Now, let us construct the Abel map for nodal curves.  For each Weil divisor $D$ on $X$, we define its \textit{multidegree} $\deg D = (d_1, \dots, d_{r}) \in \Z^{r}$ where $r = c(X)$, $d_i = \deg \left(D\mid_{X_i}\right)$. The total degree of $D$ is the number $$|\deg D| = \sum_{i=1}^{r} d_i.$$ We denote the set of divisors of multidegree $d$ by $\Div_d(X)$. 
\begin{definition}A divisor $D$ on $X$ is called \textit{$\Sigma$-regular}, if its support does not intersect $\mathrm{supp}(\Sigma)$. 
\end{definition}
The set of $\Sigma$-regular divisors of multidegree $d$ will be denoted by $\Div_d(X, \Sigma)$. The set of all $\Sigma$-regular divisors
$$
\Div(X, \Sigma) = \bigsqcup_{d \in \Z^{c(X)}} \Div_d(X, \Sigma)
$$
is a $\Z^k$-graded Abelian group.\par Let $\pazocal M^*(X, \Sigma)$ be the set of invertible elements in $\pazocal M(X, \Sigma)$:
$$
\pazocal M^*(X, \Sigma) = \{ f \in\pazocal M(X, \Sigma) :  f(P) \neq 0 \,\, \forall \,\, P \in \mathrm{supp}(\Sigma), \,\, f\mid_{X_j} \not\equiv 0 \,\,\forall\,\,  1 \leq j \leq c(X) \}.
$$
For each $f \in \pazocal M^*(X, \Sigma)$, its divisor $(f)$ is a $\Sigma$-regular divisor. 
\begin{definition}Divisors of the form $(f)$ where $f \in \pazocal M^*(X, \Sigma)$ will be called \textit{$\Sigma$-principal}. Two $\Sigma$-regular divisors are $\Sigma$-\textit{linearly equivalent} if their difference is a $\Sigma$-principal divisor. 
\end{definition}
 We denote the space of $\Sigma$-principal divisors by $\PDiv(X, \Sigma)$. For two $\Sigma$-equivalent divisors, we write:
$$D_1 \stackrel{  \Sigma}{\sim} D_2.$$
Let $D$ be a $\Sigma$-regular divisor of multidegree $0$. Then $D$ can be written as
$$
D = \sum\nolimits_{i=1}^{c(X)}(D_i^+ - D_i^-)
$$
where $D_i^{\pm}$ are effective divisors on $X_i$, and $\deg D_i^{+} = \deg D_i^{-}$. For a $\Sigma$-regular differential \nolinebreak$\omega$, we set
$$
\int_D \omega = \sum\nolimits_{i=1}^{c(X)}\int_{D_i^-}^{D_i^+}\omega.
$$
This integral is defined up to periods of $\omega$, hence we obtain a map $$\pazocal A_\Sigma \colon \Div_0(X,\Sigma) \to \Jac(X, \Sigma)$$ which is the analogue of the usual Abel map. \par
%The natural restriction map $r \colon   \left(\Omega^1(X,\Sigma)\right)^* \to  \left(\Omega^1(X)\right)^*$ maps $L(X, \Sigma)$ to $L(X)$ where $L(X)$ is the period lattice of $X$. Therefore, there exists a natural epimorphism $\pi^J \colon \Jac(X,\Sigma) \to \Jac(X) $. The notation $\pi^J$ reflect the fact that this map is backward with respect to the projection $\pi \colon X \to X / \Sigma$.
%a unique map $\pi$ which makes the following diagram commutative:
%\begin{align}\label{prMap}
%       				 \begin{CD}
%           			 	( \Omega^1(C,\Sigma))^*@>  i^* >>  (\Omega^1(C))^* \\
%            				@VV V  @VVV \\
%          				\Jac(C, \Sigma) @> \pi>>\Jac(C)
%       			 \end{CD}
%   			 \end{align}
%			 where the vertical arrows are natural projections.
%The following statement is straightforward:
\begin{statement}\label{comd}
The following diagram is commutative:
\begin{align*}
       				 \begin{CD}
           			 	 \Div_0(X,\Sigma)  @>   >> \Div_0(X) \\
            				@VV \pazocal A_\Sigma V  @VV \pazocal A V \\
          				\Jac(X, \Sigma) @> \pi^*>>\Jac(X)
       			 \end{CD}
   			 \end{align*}
			 where the upper arrow is the natural inclusion, and $\pazocal A$ is the usual Abel map.
\end{statement}
\begin{proof}Obvious. \end{proof}
\subsubsection{Abel theorem and generalized Picard group}
Let   $\PDiv(X) $ be the group of principal divisors on $X$, and let $D \in \PDiv(X) $. Then we can find a meromorphic function \nolinebreak$f$ such that $D = (f)$. 
This function is defined up to a factor which is constant on each component of $X$. To each edge $e_i$ of the dual graph $\Gamma(X, \Sigma)$ we assign a number $$r_i = f(P_i^+) / f(P_i^-) \in \Complex^*.$$ The numbers $\{r_i\}$ define a $1$-cocycle on the dual graph. Since $f$ is defined up to a locally constant factor, this cocycle is defined up to a coboundary. Therefore, to each principal divisor we can assign a cohomology class. Denote this class by $\pazocal R(D)$. We have a mapping
$$
\pazocal R \colon \PDiv(X) \cap  \Div(X, \Sigma) \to \Hom^1(\Gamma(X,\Sigma),  \Complex^*).
$$
and an exact sequence
\begin{align}\label{RDiag}
\begin{CD}
           			 	0 @>>>\PDiv(X, \Sigma) @> >>\PDiv(X) \cap  \Div(X, \Sigma) @> \pazocal R >>\Hom^1(\Gamma,  \Complex^*) @>>> 0.
       			 \end{CD}
   			 \end{align}
\begin{statement}\label{conv}
The following diagram commutes:
%$$
%     				 \begin{CD}
%         			 	\PDiv(X) \cap  \Div_0(X, \Sigma) @>  \pazocal A_\Sigma >>   \Jac(X,\Sigma) \\
%            				@VV r V  @A\Res_\Sigma^*AA \\
%         				\Hom^1(\Gamma(X,\Sigma),  \Complex^*) @> \pi>> \Hom^1(\Gamma(X,\Sigma), \Complex) 
%   			 \end{CD}
%			 $$
%\begin{figure}[h]
{
\begin{equation}
\label{convD}
\centering
\begin{tikzcd}
%& \PDiv(C) \cap  \Div_0(C, \Sigma) \arrow{dl}{\pi_{i}}\arrow{dr}[swap]{\pi_{j}} & \\
\PDiv(X) \cap  \Div(X, \Sigma) \arrow{r}{\pazocal A_\Sigma} \arrow{d}{\pazocal R} & \Jac(X, \Sigma)\\
 \Hom^1(\Gamma,  \Complex^*)  \arrow{r}{\ln_*} &  \Hom^1(\Gamma,  \Complex / 2\pi\mathrm{i}\Z) \arrow{u}{Z^*}
\end{tikzcd}
\end{equation}
}
where $\ln_*$ is the map induced by $\ln \colon  \Complex^* \to \Complex/2\pi\mathrm{i}\Z$.
%\begin{align*}
%       				 \begin{CD}
%           			 	\PDiv(X) \cap  \Div_0(X, \Sigma)  @> \pazocal A_\Sigma  >>  \Jac(X, \Sigma)  \\
%            				@VV R V  @AA Z^* A  \\
%          				\Hom^1(\Gamma(X,\Sigma),  \Complex^*) @>>>\Hom^1(\Gamma(X,\Sigma),  \Complex/ \Z)
%       			 \end{CD}
%   			 \end{align*}
%\end{figure}
%Let $D$ be a $\Sigma$-regular principal divisor. Then
%$$
%i(\hat h(D)) =\pazocal A(D).
%$$
\end{statement}
\begin{proof}
Let $D = (f) \in \PDiv(X) \cap  \Div(X, \Sigma) $. Choose a path $\gamma$ joining $\infty$ and $0$ on the Riemann sphere such that $f^{-1}(\gamma)$ does not intersect $\mathrm{supp}(\Sigma)$. Let $\omega$ be a $\Sigma$-regular differential. Then
\begin{align*}
\langle \pazocal A_\Sigma(D), \,&\omega \rangle = \int_{f^{-1}(\gamma)} \omega = \int_\gamma \tr_{\!\!f}\, \omega =  -\int_0^\infty \sum\nolimits_{i=1}^{|\Sigma|}\left( \frac{1}{z - P_i^+} -  \frac{1}{z - P_i^-}\right)\Res_{P_i^+}\,f \, \diff z=  \\ &=\sum\nolimits_{i=1}^{|\Sigma|} \ln \frac{f(P_i^+)}{f(P_i^-)}\,\Res_{P_i^+}\,\omega = \langle\ln_* \pazocal R(D), Z(\omega) \rangle =  \langle Z^*\left(\ln_* \pazocal R(D)\right), \omega \rangle. %=\ \langle \hat h(D), h(\omega) \rangle  =\langle h^*( \hat h(D)), \omega \rangle .
\end{align*}

\end{proof}

\begin{statement}[Abel theorem for nodal curves]\label{at}
A $\Sigma$-regular divisor $D$ of multidegree zero is $\Sigma$-principal if and only if $\pazocal A_\Sigma(D) = 0$.
\end{statement}
\begin{proof}
Let $D$ be a $\Sigma$-principal divisor. Then $ \pazocal R(D) = 0$, so $\pazocal A_\Sigma(D) =Z^*\left(\ln_* \pazocal R(D)\right) = 0$. Vice versa, let $\pazocal A_\Sigma(D) = 0$. Applying Proposition \ref{comd}, we conclude that $\pazocal A(D) = 0$. So, by the standard Abel theorem, $D$ is principal. Since  $\pazocal A_\Sigma(D) = 0$, and $Z^*$ is injective,  we conclude that $\pazocal R(D) = 0$, so the divisor $D$ is $\Sigma$-principal.
\end{proof}
%\subsection{Picard group and Abel-Jacobi theorem}
Using diagrams \eqref{9diag} and \eqref{RDiag}, the diagram \eqref{convD} is reduced to
{
\begin{equation}
\centering
\begin{tikzcd}
%& \PDiv(C) \cap  \Div_0(C, \Sigma) \arrow{dl}{\pi_{i}}\arrow{dr}[swap]{\pi_{j}} & \\
\PDiv(X) \cap  \Div(X, \Sigma) / \PDiv(X,\Sigma)  \arrow{r}{\pazocal A_\Sigma} \arrow{d}{\pazocal R} & \Ker \pi^*\\
 \Hom^1(\Gamma,  \Complex^*)  \arrow{r}{\ln_*} &  \Hom^1(\Gamma,  \Complex / 2\pi\mathrm{i}\Z) \arrow{u}{Z^*}
\end{tikzcd}
\end{equation}
}
where $R$, $\ln_*$, and $Z^*$ are isomorphisms. Therefore, $\pazocal A_\Sigma$ is also an isomorphism.
\begin{definition}
The generalized \textit{generalized Picard group} is
 $$\Pic_0(X,\Sigma) =   \Div(X, \Sigma) /  \PDiv(X, \Sigma).  $$
 \end{definition}
The group $\Pic(X,\Sigma)  $ is a $\Z^k$-graded Abelian group:
$$
\Pic(X,\Sigma) = \bigsqcup_{d \in \Z^{c(X)}} \Pic_d(X,\Sigma)
$$
 where the multidegree $d$ generalized Picard variety $\Pic_d(X,\Sigma)$ is the set of $\Sigma$-regular divisors of multidegree $d$ modulo linear equivalence. We denote the $\Sigma$-linear equivalence class of a divisor $D$ by $[D]_\Sigma$, or just $[D]$.
%\begin{statement}
%$\Jac(C,\Sigma)$
%\end{statement}
\begin{statement}[Abel-Jacobi theorem for nodal curves]
The Abel map is an isomorphism between $\Pic_0(X,\Sigma) $ and $\Jac(X,\Sigma)$. 
\end{statement}
\begin{proof}
By Proposition \ref{at}, the Abel map $\pazocal A_\Sigma \colon \Pic_0(X,\Sigma)  \to \Jac(X,\Sigma)$ is injective. Let us prove that it is surjective. Take $x \in \Jac(X,\Sigma)$. Then $\pi^*(x) \in \Jac(X)$, and by the classical Abel-Jacobi theorem, there exists $D \in   \Div_0(X)$ such that $\pazocal A(D) = \pi^*(x)$. As it is easy to see, $D$ may be chosen to be $\Sigma$-regular.
Then $x - \pazocal A_\Sigma(D) \in \Ker \pi^*$. As it is proved above, the Abel map  is an isomorphism between $ (\PDiv(X) \cap  \Div(X, \Sigma)) /  \PDiv(X)   $ and $\Ker \pi^*$, so there exists $D' \in \Div_0(X, \Sigma)$ such that  $\pazocal A_\Sigma(D') = x - \pazocal A_\Sigma(D) $, i.e.  $x = \pazocal A_\Sigma(D'+D) $, q.e.d.
\end{proof}

 The variety $\Pic_d(X,\Sigma)$ is thus a principal homogeneous space of the group $$\Pic_0(X,\Sigma)  \simeq \Jac(X,\Sigma)$$ for each multidegree $d$. In particular, $\Pic_d(X,\Sigma)$ has a canonical affine structure and its tangent space at each point can be naturally identified with $\Omega^1(X,\Sigma)^\ast$.
 \subsubsection{Partial normalizations and subcurves}\label{pnsc}
 Let $\Lambda \subset \Sigma$. Then $X / \Lambda$ is a \textit{partial normalization} of the curve $X/\Sigma$. Each $\Sigma$-regular divisor is also a $\Lambda$-regular divisor, hence there is a natural inclusion map $ \Div(X,\Sigma)  \to \Div(X,\Lambda) $. \pagebreak[3]
 \begin{statement}
There exists a unique graded epimorphism $\pi_{\Lambda}^*$ which makes the following diagram commutative:
\begin{align*}
       				 \begin{CD}
           			 	 \Div(X,\Sigma) @>   >>  \Div(X,\Lambda)  \\
            				@VV V  @VVV \\
          				 \Pic(X,\Sigma)  @> \pi_{\Lambda}^*>> \Pic(X,\Lambda) 
       			 \end{CD}
   			 \end{align*}
			 where the upper arrow is the natural inclusion, and vertical arrows are natural projections.
 \end{statement}
 \begin{proof}Obvious. \end{proof}
The notation $ \pi_{\Lambda}^*$ reflects the fact that this map is backward with respect to the partial normalization map $\pi_{\Lambda} \colon X/\Lambda \to X/\Sigma$.
%In this way, the Picard group of a nodal curve has a structure of a fiber bundle over the Picard group of any partial normalization of this curve.

Now, let us discuss subcurves. Let $I = \{i_1, \dots, i_p\} \subset \{ 1, \dots, c(X)\}$, and let $$ X_I = \bigsqcup_{i \in I} X_i.$$ Let also $$ \Sigma_I= \{ \{P_i^+,P_i^-\} \in \Sigma \mid P_i^+ \in  X_I, P_i^- \in  X_I \}.$$ Then $ X_I /  \Sigma_I$ is a \textit{subcurve} of  $ X /  \Sigma$. Subcurves of $X/\Sigma$ correspond to complete subgraphs of its dual graph. If  $ X_I /  \Sigma_I$ is a sub-curve of  $ X /  \Sigma$, then for each multidegree $d = (d_1, \dots , d_{c(X)})$ on $X$ there is a natural restriction map $ \Div_d(X,\Sigma)  \to \Div_{d_I}(X_I,\Sigma_I) $ where $d_I = (d_{i_1}, \dots, d _{i_p})$.
 \begin{statement}
There exists a unique epimorphism $i_{I}^*$ which makes the following diagram commutative:
\begin{align*}
       				 \begin{CD}
           			 	 \Div_d(X,\Sigma) @>   >>  \Div_{d_I}(X_I,\Sigma_I)  \\
            				@VV V  @VVV \\
          				 \Pic_d(X,\Sigma)  @> i_{I}^*>> \Pic_{d_I}(X_I,\Sigma_I) 
       			 \end{CD}
   			 \end{align*}
			 where the upper arrow is the natural restriction map, and vertical arrows are natural projections.
 \end{statement}
 \begin{proof}Obvious. \end{proof}
 The notation $ i_{I}^*$ reflects the fact that this map is backward with respect to the natural inclusion map $i_{I} \colon X_I/\Sigma_I \to X/\Sigma$.\par
 %This map sends $ \PDiv(C,\Sigma) $ to   $ \PDiv(C,\Sigma') $ .
%   Let $ X_I/\Sigma_I \subset X/\Sigma $ be a subcurve. Taking $I'= \{ 1, \dots, c(X)\} \setminus I$, we obtain the complimentary subcurve $X_{I'} / \Sigma_{I'}$.
%Let us define the number $\kappa(I) = |\Sigma| - |\Sigma_I| - |\Sigma_{I'}|$. This number equals the geometric number of points in the intersection $X_I / \Sigma_I \cap X_{I'} / \Sigma_{I'}$. 

\subsubsection{Riemann's inequality and effective divisors}
 Let $D \in \Div(X,\Sigma)$, and let
$$
\mathrm{L}(D,\Sigma) = \{ f \in \pazocal M(X, \Sigma) \mid \mathrm{ord}_P\,f \geq -D(P) \,\,\forall\,\, P \in X \}
$$
where $\mathrm{ord}_P\,f $ is the order of $f$ at the point $P$, and we set $\mathrm{ord}_P\,f =\infty$ if $P \in X_i$ and \nolinebreak $f\mid_{X_i} \equiv 0$. Obviously, the set $\mathrm{L}(D,\Sigma) $ is a vector space. 
\begin{statement}[Riemann's inequality for nodal curves] For each $D \in \Div(X,\Sigma)$, the following inequality holds $$\dim \mathrm{L}(D, \Sigma) \geq |\deg D| - g(X,\Sigma) + c(X,\Sigma).$$
%where $c$ is the number of connected components of $X / \Sigma$.
\end{statement}
\begin{proof}
Let $\mathrm{L}(D) = \mathrm{L}(D,\emptyset)$. We have
$$
\dim \mathrm{L}(D) = \sum_{i=1}^{c(X)} \dim \mathrm{L}(D\mid_{X_i}) \geq  \sum_{i=1}^{c(X)}\left( \deg D\mid_{X_i} -g(X_i) + 1\right) = |\deg D| + c(X)  - g(X).
$$
Consider a map $\delta \colon \mathrm{L}(D) \to \Complex^{|\Sigma|}$ given by
$
\delta(f) = (f(P_1^+) - f(P_1^-), \dots).
$
We have $\mathrm{L}(D, \Sigma) = \Ker \delta$, so
$$\dim \mathrm{L}(D, \Sigma) = \dim \mathrm{L}(D) - \dim \Imm \delta \geq \dim \mathrm{L}(D) - |\Sigma| =  |\deg D| - g(X,\Sigma) + c(X,\Sigma).$$
\end{proof}
Let $D$ and $D'$ be $\Sigma$-linearly equivalent divisors. Then it easy to see that there exists an isomorphism $\phi \colon \mathrm{L}(D,\Sigma) \to \mathrm{L}(D', \Sigma)$. This allows us to define the set
$$
W_d(X,\Sigma) =\{ [D] \in  \Pic_d(X,\Sigma) : \mathrm L(D,\Sigma) \neq 0\}.
$$
By Riemann's inequality, we have $W_d(X,\Sigma) =  \Pic_d(X,\Sigma)$ if $|d| \geq g(X,\Sigma) - c(X,\Sigma) + 1$. However, if the curve is reducible, then it may happen that
$W_d(X,\Sigma) =  \Pic_d(X,\Sigma)$ even if $|d| \leq g(X,\Sigma) - c(X,\Sigma)$. In particular, if $|d| = g(X,\Sigma) - c(X,\Sigma)$, then the set $W_d(X,\Sigma) $ is a proper subset of $ \Pic_d(X,\Sigma)$ if and only if $d$ satisfies the so-called \textit{semistability} condition.
 \begin{definition}Let $d$ be a multidegree of total degree $g(X,\Sigma) - c(X,\Sigma)$. Then $d$ is called \textit{semistable}
 %\footnote{Compare with the notion of semistable \cite{Alexeev,Caporaso} and balanced \cite{Caporaso2} multidegree.}
  if for each subcurve $ X_I /  \Sigma_I \subset  X /  \Sigma$ we have
\begin{align}\label{LINEQ}
\left\lvert d_I \right\rvert \geq g( X_I, \Sigma_I) -  c( X_I, \Sigma_I).
\end{align}
 \end{definition}

%A condition equivalent to semistability first appeared in Beauville \cite{beauville}. The term \textit{semistable multidegree} is due to Alexeev \cite{Alexeev}. %See also Caporaso \cite{Caporaso}.

\begin{statement}\label{semistabMain}
Let $d$ be a multidegree of total degree $g(X,\Sigma) - c(X,\Sigma)$. Then
\begin{longenum}
\item if $d$ is semistable, then $W_d(X,\Sigma) $ has a positive codimension in $\Pic_d(X,\Sigma)$;
\item  if $d$ is not semistable, then $W_d(X,\Sigma)  = \Pic_d(X,\Sigma)$.
\end{longenum}
\end{statement}
This result is due to Beauville \cite{beauville} and Alexeev \cite{Alexeev}. We will get a proof of this statement as a by-product of our further considerations.\par
% which maps $\mathrm L^{(r)}(D,\Sigma)$ to $\mathrm L^{(r)}(D', \Sigma)$ and $\mathrm L^{(n)}(D,\Sigma)$ to $\mathrm L^{(n)}(D', \Sigma)$.
%We note that if $D_1 \stackrel{\Sigma} {\sim}D_2$, then the spaces $\mathrm{L}(D_1,\Sigma)$ and $\mathrm{L}(D_2, \Sigma)$ are naturally isomorphic. We define $$W_{d}(X, \Sigma) = \{ [D] \in  \Pic_d(X,\Sigma): \dim \mathrm{L}(D,\Sigma) > 0 \}.$$
Similarly to $W_d(X,\Sigma)$, we define $E_{d}(X, \Sigma) \subset \Pic_d(X,\Sigma)$ as the set of those classes \nolinebreak of divisors which are representable by effective divisors. Obviously, 
we have $E_{d}(X, \Sigma) \subset W_d(X,\Sigma)$. Moreover, these two sets are equal for non-singular connected curves.\par
 If we define $\mathring X_i = X_i \setminus (X_i \cap \mathrm{supp}(\Sigma))$, then the set $E_{d}(X, \Sigma)$ can be described as the image of the map
$$
 (\mathring X_1)^{d_1} \times \dots \times (\mathring X_k)^{d_k} \to \Pic_d(X,\Sigma)
$$
which maps a collection of points to the corresponding effective divisor. This description makes it obvious that $E_{d}(X, \Sigma)$ has a positive codimension in $\Pic_d(X,\Sigma)$ if $|d| < g(X,\Sigma)$. However, for reducible curves, the set $E_{d}(X, \Sigma)$ can have a positive codimension even if $|d| \geq g(X,\Sigma)$. This motivates us to give the following definition.
 \begin{definition}Let $d$ be a multidegree of total degree $g(X,\Sigma)$. We say that $d$ is \textit{uniform}
 %\footnote{Compare with the notion of semistable \cite{Alexeev,Caporaso} and balanced \cite{Caporaso2} multidegree.}
  if for each subcurve $ X_I /  \Sigma_I \subset  X /  \Sigma$ we have
$$
\left\lvert d_I \right\rvert \geq g( X_I, \Sigma_I).
$$ 
\end{definition}
%The significance of uniform multidegrees is motivated by the following statement.
\begin{statement}\label{density}
Let $d$ be a multidegree of total degree $g(X,\Sigma)$. \begin{longenum}
\item if $d$ is uniform, then $E_d(X,\Sigma) $ is dense in $\Pic_d(X,\Sigma)$;
\item  if $d$ is not uniform, then $E_d(X,\Sigma)$ has positive codimension in  $\Pic_d(X,\Sigma)$.
\end{longenum} \end{statement}
We shall postpone the proof of this result until we have obtained several preliminary statements. Let us consider the following decomposition of the space $\mathrm{L}(D, \Sigma)$:
$$
\mathrm{L}(D, \Sigma) = \mathrm L^{(*)}(D, \Sigma) \sqcup \mathrm L^{(r)}(D, \Sigma) \sqcup \mathrm L^{(n)}(D, \Sigma) 
$$
where 
\begin{align*}
 \mathrm L^{(*)}(D, \Sigma)  =\mathrm{L}(&D, \Sigma)  \cap  \pazocal M^*(X,\Sigma), \\  \mathrm L^{(r)}(D, \Sigma) = \{  f \in \mathrm{L}(D,\Sigma) &:  f\mid_{X_i} \equiv 0 \mbox{ for some } 1 \leq i \leq k \},\\
 \mathrm L^{(n)}(D, \Sigma) = \{  f \in \mathrm{L}(D,\Sigma) \setminus \mathrm L^{(r)}(D,&\, \Sigma)  :  f(P_i^\pm) = 0 \mbox{ for some } 1 \leq i \leq |\Sigma| \}.
\end{align*}
\begin{statement}\label{effect}Let $D \in \Div_d(X,\Sigma)$. Then $D \in E_d(X,\Sigma)$ if and only if $L^{(*)}(D, \Sigma) \neq \emptyset$.
\end{statement}
\begin{proof} Obvious.\end{proof}
It is easy to see that if $D$ and $D'$ are $\Sigma$-linearly equivalent, then the aforementioned isomorphism $\phi \colon \mathrm{L}(D,\Sigma) \to \mathrm{L}(D', \Sigma)$ maps $\mathrm L^{(r)}(D,\Sigma)$ to $\mathrm L^{(r)}(D', \Sigma)$ and $\mathrm L^{(n)}(D,\Sigma)$ to $\mathrm L^{(n)}(D', \Sigma)$.
Therefore, for each multi-degree $d$, we can define the sets
\begin{align*}
R_d(X,\Sigma) &=\{ [D] \in  \Pic_d(X,\Sigma) : \mathrm L^{(r)}(D,\Sigma) \neq 0\}, \quad \\ N_d(X,\Sigma) &=\{ [D] \in  \Pic_d(X,\Sigma) : \mathrm L^{(n)}(D,\Sigma) \neq \emptyset\}. 
%Z_d(C,\Sigma) = \{ D \in  \Pic_d(C,\Sigma) \mid \exists \, f \in \mathrm{L}(D,\Sigma), i \in \{1,\dots, k\}  \mid f\mid_{C_i} \equiv 0\}. 
\end{align*}
We have
$
W_d(X,\Sigma) = E_d(X,\Sigma) \cup R_d(X,\Sigma) \cup N_d(X,\Sigma).
$
We note that $R_d$ is empty for irreducible curves, while $N_d$ is empty for non-singular curves. %For connected non-singular curves, we have $W_d = E_d$.
\begin{statement}\label{zd}Let $d$ be a multidegree of total degree $|d| \leq g(X,\Sigma)$. Then $N_d(X,\Sigma)$ has positive codimension in   $\Pic_d(X,\Sigma)$.
\end{statement}
\begin{proof}
Let $\Lambda \subset \Sigma$, $S = \mathrm{supp}(\Sigma \setminus \Lambda)$ and let $D_\Lambda = \sum_{P \in S} P$. %assume that $   \Sigma \setminus \Lambda = \{\{P_{i_1}^+, P_{i_1}^-\}, \dots, \{P_{i_p}^\pm\}\} $. Let
%$
%D_{\Lambda} =  P_{i_1}^+ + P_{i_1}^- +  \dots + P_{i_p}^+ + P_{i_p}^- \in  \Div(X,\Lambda)
%$.
 Let also $e(\Lambda) = d - \deg D_\lambda$.  We claim that
$$
N_d(X,\Sigma) \subset \bigcup_{\Lambda \subsetneq \Sigma} \left(\pi_\Lambda^*\right)^{-1}\left(E_{e(\Lambda)}(X, \Lambda) +  [D_\Lambda]\right).
$$
%%Since $|e(\Lambda)| = |d| - |\Sigma \setminus \Lambda|$
Indeed, let $[D] \in N_d(X,\Sigma)$. Then there exists $f \in L^{(n)}(X,\Sigma)$. Assume that $f$ vanishes at points $P_{i_1}^\pm, \dots, P_{i_p}^\pm$, and let $$\Lambda = \Sigma \setminus \{\{P_{i_1}^+, P_{i_1}^-\}, \dots, \{P_{i_p}^+, P_{i_p}^-\}\}. $$ \pagebreak[4] We have $f \in L^{(*)}(D - D_\Lambda, \Lambda)$, which implies that $[D] \in \left(\pi_\Lambda^*\right)^{-1}\left(E_{e(\Lambda)}(X, \Lambda) +  [D_\Lambda]\right),$ q.e.d.\par\smallskip Further, we have $$|e(\Lambda) | = |d| - 2p < |d| - p \leq g(X,\Sigma) - p \leq g(X, \Lambda),$$ so $E_{e(\Lambda)}(X, \Lambda)$ has positive codimension in $\Pic_{e(\Lambda)}(X, \Lambda)$, which implies the proposition.
%%
%%so $[D'] \in W_e(C,\Sigma')$ where $e = \deg D'$. We have
%%$
%%|e| = |d| - 2p.
%%$
%%On the other hand, $$g(C,\Sigma') \geq g(C,\Sigma) - p > |e|,$$
%%so $W_e(C,\Sigma')$ is of positive codimension in $\Pic_e(C,\Sigma')$, and $Z_d(C,\Sigma)$ is of positive codimension in $ \Pic_d(C,\Sigma) $.
\end{proof}
%Let $J =  \{1, \dots, c(X)\}$, and let $I  \subset J$. Let also $ I' = J \setminus I$. Then we can consider the subcurve $X_I/\Sigma_I$, and a complimentary subcurve $X_{I'}/\Sigma_{I'}$. Let $s(I) = |\Sigma| - |\Sigma_I| - |\Sigma_I'| = |X_I/\Sigma_I \cap X_{I'}/\Sigma_{I'}|$.
   Let $ X_I/\Sigma_I \subset X/\Sigma $ be a subcurve. Taking $I'= \{ 1, \dots, c(X)\} \setminus I$, we obtain the complimentary subcurve $X_{I'} / \Sigma_{I'}$.
Let us define the number $\kappa(I) = |\Sigma| - |\Sigma_I| - |\Sigma_{I'}|$ which is equal to the geometric number of points in the intersection $X_I / \Sigma_I \cap X_{I'} / \Sigma_{I'}$. 
 \begin{definition}Let $d$ be a multidegree. We say that $d$ is R-semistable
 %\footnote{Compare with the notion of semistable \cite{Alexeev,Caporaso} and balanced \cite{Caporaso2} multidegree.}
  if for each proper non-empty subcurve $ X_I /  \Sigma_I \subset  X /  \Sigma$ we have
\begin{align}\label{RINEQ}
\left\lvert d_{I}\right\rvert \leq g(X_I,\Sigma_I) -c(X_I, \Sigma_I)+  \kappa(I).
\end{align}
\end{definition}
\begin{statement}\label{yd}Let $d$ be a multidegree. Then
\begin{longenum}
\item if $d$ is $R$-semistable, then $R_d(X,\Sigma) $ has positive codimension in $\Pic_d(X,\Sigma)$;
\item  if $d$ is not $R$-semistable, then $R_d(X,\Sigma)  = \Pic_d(X,\Sigma)$.
\end{longenum}
\end{statement}
\begin{proof} Assume that $d$ is $R$-semistable.
 Let $ X_I/\Sigma_I \subset X/\Sigma $ be a subcurve, and let  $X_{I'} / \Sigma_{I'}$ be the complimentary subcurve. We have the following decomposition of the dual graph:
$$
\Gamma(X,\Sigma) = \Gamma(X_I,\Sigma_I) \sqcup  \Gamma(X_{ I'},\Sigma_{I'} ) \sqcup \{e_{j_1}, \dots, e_{j_\kappa}\}
$$
where $\kappa = \kappa(I)$.
Without loss of generality, we assume that $P_{j_1}^+, \dots, P_{j_\kappa}^+ \in X_I$, and $P_{j_1}^-, \dots, P_{j_\kappa}^- \in X_{ I'}$. Let $$D_{I} = P_{j_1}^+ + \dots + P_{j_q}^+ \in \Div(X_I, \Sigma_I),$$
and let $e(I) = d_I - \deg D_I$. We claim that
$$
R_d(X,\Sigma) \subset \bigcup_{\substack{I \subset J, \\ I \neq \emptyset}} \left(i^*_{I}\right)^{-1}\left((E_{e(I)}(X_I, \Sigma_I) \cup N_{e(I)}(X_I, \Sigma_I))+  [D_{I}]\right)
$$
where $J =  \{ 1, \dots, c(X)\} $.
Indeed, let $[D] \in R_d(X,\Sigma)$. Then there exists $f \in \mathrm L^{(r)}(D,\Sigma)$, $f \neq 0$. Assume that $f \equiv 0$ on irreducible components $X_{i_1}, \dots, X_{i_p}$. Let $I'= \{i_1, \dots i_p\}$, and let $I = J \setminus  I'$. Let also $f_I = f\mid_{X_I}$. Then $$f_I \in L^{(*)}(D\mid_{X_I} - D_{I}, \Sigma_I) \cup L^{(n)}(D\mid_{X_I} - D_{I}, \Sigma_I) ,$$ therefore  $[D] \in \left(i^*_{I}\right)^{-1}\left((E_{e(I)}(X_I, \Sigma_I) \cup Z_{e(I)}(N_I, \Sigma_I))+  [D_{I}]\right),$ q.e.d.\par \smallskip
Further, we have
%\begin{align*}
%g(X, \Sigma) = g(X_I,\Sigma_I) &+ g(X_{I'}, \Sigma_{I'}) + |\deg D_{I}| + c(X,\Sigma)- c(X_I,\Sigma_I) - c(X_{I'},\Sigma_{I'}) 
%%< \\ &<  g(X_I,\Sigma_I) + g(X_{I'}, \Sigma_{I'}) + |\deg D_{I}| .
%\end{align*}
%Setting $d_{ I'} = d\mid_{X_{ I'}}$, we have
$$
|e(I)| =  |d_I| - |\deg D_I| = |d_I| - \kappa(I) < g(X_I,\Sigma_I),
%g(X,\Sigma) - |d_{ I'}|  - |\deg D_{I}| <  g(X_I,\Sigma_I) + g(X_{ I'}, \Sigma_{ I'})  - |d_{ I'}|.
$$
%Since $d$ is uniform, $g(X_{ I'}, \Sigma_{ I'})  \leq |d_{ I'}|$, so $|e(I)| < g(X_I,\Sigma_I)$. 
so $E_{e(I)}(X_I, \Sigma_I)$ and $N_{e(I)}(X_I, \Sigma_I)$ have positive codimension in $\Pic_{e(I)}(X_I, \Lambda_I)$, which implies that $R_d(X,\Sigma) $ has positive codimension in $\Pic_d(X,\Sigma)$.\par\smallskip
Now, assume that $d$ is not $R$-semistable. Then there exists a subcurve such that $$\left\lvert d_{I}\right\rvert \geq g(X_I,\Sigma_I) -  c(X_I,\Sigma_I)  +  \kappa(I) + 1. $$ Let $[D] \in \Pic_d(X,\Sigma)$. We shall prove that $\mathrm L^{(r)}(D, \Sigma) \neq 0$.
We have $$|\deg (D\mid_{X_I} -  D_I)| \geq g(X_I,\Sigma_I) -  c(X_I,\Sigma_I) + 1,$$ so by Riemann's inequality there exists $f \neq 0$ such that $f \in \mathrm L(D\mid_{X_I} -  D_I, \Sigma_I)$. Take $\wave f = f$ for $P \in X_I$, and $\wave f = 0$ for $P \notin X_I$. Then $f \in \mathrm L^{(r)}(D, \Sigma)$, q.e.d.
\end{proof}
\begin{proof}[Proof of Proposition \ref{semistabMain}]
Using the obvious formula
\begin{align}\label{genusOfUnion}
g(X, \Sigma) =  g(X_I,\Sigma_I) + g(X_{I'}, \Sigma_{I'}) +\kappa(I) + c(X,\Sigma) - c(X_I,\Sigma_I) - c(X_{I'},\Sigma_{I'}),
\end{align}
we show that if $|d| = g(X,\Sigma) - c(X,\Sigma)$, then inequality \eqref{LINEQ} for a subcurve $X_I/\Sigma_I$ is equivalent to  inequality \eqref{RINEQ} for the complimentary subcurve $X_{I'}/\Sigma_{I'}$, and vice versa. So, for  $|d|=g(X,\Sigma) - c(X,\Sigma)$, the notions ``semistable'' and ``R-semistable'' coincide (cf. \cite{Caporaso}, Remark 1.3.3). Therefore, if the multidegree $d$ is semistable, then it is R-semistable, and the set $R_d(X,\Sigma) $ has positive codimension in $\Pic_d(X,\Sigma)$. At the same time, since $|d| < g(X,\Sigma)$, the sets $E_d(X,\Sigma) $ and $N_d(X,\Sigma) $ also have positive codimension in $\Pic_d(X,\Sigma)$, thus the same is true for $W_d(X,\Sigma)$.\par
Vice versa, if $d$ is not semistable, then it is not R-semistable, therefore $R_d(X,\Sigma)  = \Pic_d(X,\Sigma)$, and $W_d(X,\Sigma)  = \Pic_d(X,\Sigma)$.
\end{proof} \pagebreak[4] 
%\begin{statement}\label{yd2}Let $X /\Sigma$ be a connected nodal curve, and let $d$ be a uniform multidegree of total degree $g(X,\Sigma)$. Then the set $R_d(X,\Sigma)$ has a positive codimension in  $ \Pic_d(X,\Sigma)$.
%\end{statement}
%\begin{proof}
%Let us show that for connected curves uniformity implies \eqref{rdineq}.  Let $ X_I/\Sigma_I \subset X/\Sigma $ be a proper non-empty subcurve. We have
%\begin{align*}
%g(X, \Sigma) =  g(X_I,\Sigma_I) &+ g(X_{I'}, \Sigma_{I'}) +\kappa(I) + 1 - c(X_I,\Sigma_I) - c(X_{I'},\Sigma_{I'}) < \\
%&<  g(X_I,\Sigma_I) + g(X_{I'}, \Sigma_{I'}) +\kappa(I),
%\end{align*}
%Using uniformity condition, we get $$|d_I| = |d| - |d_{I'}| = g(X,\Sigma) -  |d_{I'}| \leq g(X,\Sigma) -  g(X_{I'},\Sigma_{I'}) <  g(X_I,\Sigma_I) + \kappa(I).$$
%\end{proof}
\begin{proof}[Proof of Proposition \ref{density}]
Assume $d$ uniform. %Let $D \in  \Pic_d(X,\Sigma) \setminus E_d(X,\Sigma)$. 
%We shall prove that the set of such divisors is of positive codimension. 
If the curve $X / \Sigma$ can be represented as the disjoint union of two proper subcurves $X_1 / \Sigma_1$ and $X_{I'} / \Sigma_{I'}$, then the multidegrees $d_I = d\mid_{X_I}$ and  $d_{I'} = d\mid_{X_{I'}}$ are also uniform. Furthermore, we have $$\Pic_d(X,\Sigma) = \Pic_{d_I}(X_I,\Sigma_I) \times \Pic_{d_{I'}}(X_{I'},\Sigma_{I'}), \quad E_d(X,\Sigma) = E_{d_I}(X_I,\Sigma_I) \times E_{d_{I'}} (X_{I'},\Sigma_{I'}).$$
Therefore, without loss of generality, we may assume that $X / \Sigma$ is connected. 
Let us show that uniform multidegrees on connected curves are $R$-semistable. Using \eqref{genusOfUnion}, we get
$$
g(X, \Sigma) \leq  g(X_I,\Sigma_I) + g(X_{I'}, \Sigma_{I'}) +\kappa(I)  - c(X_{I},\Sigma_{I}).
$$
Using uniformity condition, we have $$|d_I| = |d| - |d_{I'}| = g(X,\Sigma) -  |d_{I'}| \leq g(X,\Sigma) -  g(X_{I'},\Sigma_{I'}) \leq g(X_I,\Sigma_I) +\kappa(I)  - c(X_{I},\Sigma_{I})
.$$
We conclude that $R_d(X,\Sigma) $ has positive codimension in  $\Pic_d(X,\Sigma) $. Since $N_d(X,\Sigma)$ also has positive codimension, and 
$$
 \Pic_d(X,\Sigma) = W_d(X,\Sigma) = E_d(X,\Sigma) \cup R_d(X,\Sigma) \cup N_d(X,\Sigma),
$$
we conclude that $ E_d(X,\Sigma)$ is dense in $ \Pic_d(X,\Sigma) $, q.e.d.\par
Now, assume that $d$ is not uniform. Then there exists a subcurve $X_I / \Sigma_I $ such that $|d_I| < g(X_I, \Sigma_I)$, so that the set $E_{d_I}(X_I,\Sigma_I)$ has positive codimension in $\Pic_{d_I}(X_I, \Sigma_I)$. At the same time, we have an inclusion $i^*_{I}(E_d(X,\Sigma)) \subset E_{d_I}(X_I,\Sigma_I)$, which proves that $E_d(X,\Sigma)$ has positive codimension in $ \Pic_d(X,\Sigma) $, and thus is not dense.
\end{proof}
\subsubsection{On more general curves}
We can generalize the discussion of this section to a slightly more general class of curves which we call \textit{generalized nodal curves}. Let $X = X_1 \sqcup \ldots \sqcup X_k$ be a disjoint union of connected Riemann surfaces, and let $\Sigma = \{ \pazocal P_1, \dots, \pazocal P_\sigma\}$ be a set of pairwise disjoint finite subsets of $X$. A generalized nodal curve $X / \Sigma$ is obtained from $X$ by gluing points within each $\pazocal P_i$ to a single point. The ring of meromorphic functions is defined as
$$
\pazocal M(X,\Sigma) = \{ f \in \pazocal M(X) \mid  P,Q \in \pazocal P_i \Rightarrow f(P)=f(Q) \neq \infty \,\,\forall\,\,  i = 1, \dots, \sigma\},
$$
and $\Sigma$-regular differentials are those which are holomorphic outside $\mathrm{supp}(\Sigma) =  \pazocal P_1 \cup \ldots \cup \pazocal P_\sigma$ and satisfy
$$
\sum_{P \in \pazocal P_i}\Res_P\,\omega = 0 
$$
for each $i$.
The study of such curves can be easily reduced to nodal curves. Assume that $\pazocal P_i = \{P_1, \dots, P_s\}$. Consider a Riemann sphere $Y \simeq \CP^1$ with $s$ marked points $Q_1, \dots, Q_s$. Let $ X' = X \sqcup Y$, and let $$ \Sigma' = \Sigma \cup \{ \{P_1, Q_1\}, \dots, \{P_s,Q_s\}\} \setminus \{\pazocal P_i\}.$$
The curve $X'/\Sigma'$ is ``equivalent'' to $X/\Sigma$ in the following sense: there are natural identifications
$$
\pazocal M(X,\Sigma) \simeq \{ f \in \pazocal M(X',\Sigma')\} \mid f\mid_Y \equiv \const \}, \quad \Omega^1(X,\Sigma) \simeq \Omega^1(X',\Sigma').
$$
These identifications allow to reduce the study of $X/\Sigma$ to $X'/\Sigma'$.
Repeating this operation for each $i$ such that $|\pazocal P_i| > 2$, we obtain a nodal curve, which shows that all results of this section are true for generalized nodal curves as well.
\begin{statement}\label{aggnc}
The arithmetic genus of a generalized nodal curve is given by $$ g(X,\Sigma) = g(X)+|\mathrm{supp}(\Sigma)| - |\Sigma| + c(X,\Sigma) - c(X) .$$
\end{statement}
\subsection{Proof of Theorem \ref{thm1}}\label{proof1}
\subsubsection{Preliminaries}
 Let  $X$ be the non-singular compact model of $C$. Then $ \lambda$ and $\mu$ are meromorphic functions on $X$. 
 %Together they define a projection $\pi \colon X \to \overline{ C}$ where $\overline {C}$ is the closure of $C$ in $\CP^1 \times \CP^1$. This projection is given by $P \mapsto (\lambda(P), \mu(P))$. 
 Using the simplicity of the spectrum of $J$, we conclude that \nolinebreak $\lambda$ has exactly $n$ simple poles on $X$. We denote these poles by $\infty_1, \dots, \infty_n$.  Without loss of generality, we may assume that $J = \mathrm{diag}(j_1, \dots, j_n)$, and that  $\mu\lambda^{-m}$ takes value $j_i$ at the point $\infty_i$. We also define $$D_\infty = \sum_{i=1}^n \infty_i \in \Div_n(X),$$ and $X_\infty = \mathrm{supp}(D_\infty).$ Let us also consider the projection $\pi \colon X \setminus X_\infty \to C$ given by $P \mapsto (\lambda(P), \mu(P))$, and let $X_S = \pi^{-1}(\Sing C)$. Let also 
 $$
 D_{s} = \sum_{P \in X_s} P \in \Div(X).
 $$
 
 Let us find the genus of $X$. Let $(\lambda)_R = \sum_{P \in X}(\mathrm{mult}_P \,\lambda - 1)$ be the ramification divisor of $\lambda \colon X \to \CP^1$. The following is clear.
 \begin{statement}
 We have
 $(\lambda)_R = (\diffFXi{\chi}{\mu})_0 - D_s$
 where $\chi$ is the defining polynomial of the curve $C$, and $(\,\dots)_0$ denotes the divisor of zeros.
 \end{statement}
 We conclude that
 $$
 \deg (\lambda)_R = \deg \left(\diffFX{\chi}{\mu}\right)_0 - \deg D_s = \deg \left(\diffFX{\chi}{\mu}\right)_\infty - \deg D_s
 $$
 where $(\,\dots)_\infty$ is the pole divisor. Further, from the condition $C \in \mathcal C_{spec}$, we easily conclude that
 $
\left \lvert  \deg \left(\diffFXi{\chi}{\mu}\right)_\infty \right\rvert= mn(n-1),
 $
 so, by Riemann-Hurwitz formula, we have
\begin{align}\label{geomGenus}
 g(X) = c(X) - n + \frac{1}{2}|\deg (\lambda)_R| =  c(X) - n + \frac{{mn(n-1)}}{2} - |\Sing C|.
\end{align}
 \par
We shall work with two singularizations of $X$, namely the curves $X_K$ and $X'_K$. In terms of Section \ref{nc}, they are described as follows.\par Let $\hat K = \Sing C \setminus K$, and assume that $\hat K = \{ P_1, \dots, P_\sigma\}$. Then $\pi^{-1}(P_i)$ consists of two points $ P_i^+, P_i^-$. Let us define  $$\Sigma = \{ \{P_1^+, P_1^-\}, \dots, \{P_\sigma^+, P_\sigma^-\}\}. $$ Then $ X / \Sigma = X_{K}$.
% Let us also define $D_\Sigma = \sum_{i=1}^\sigma (P_i^+ + P_i^-) \in \Div_{2\sigma}(X)$. %Obviously, we have $\lambda, \mu \in \pazocal M(X, \Sigma)$. \par%We also define $D_\Sigma = \sum_{P\in \mathrm{supp}(\Sigma)}P$.\par
Similarly, we define $\Sigma' = \Sigma \cup \{X_\infty \}$, so that $X / \Sigma' = X'_{K} $. 
\begin{statement} We have
$$
g(X, \Sigma) =  \frac{mn(n-1)}{2}  - |K| + c(X,\Sigma) - n, \quad
g(X, \Sigma') =  \frac{mn(n-1)}{2}  - |K|.
$$
\end{statement}
\begin{proof}
Use  formula \eqref{geomGenus} and Proposition \ref{aggnc}.
\end{proof}
%Without loss of generality, we may assume that $J$ is a diagonal matrix.
%Without loss of generality, we may assume that $J$ is a diagonal matrix.
\subsubsection{Construction of the map $\Phi$}\label{cons}
%Let $\overline C_p$ be the closure of $C_p$ in $\CP^1 \times \CP^1$, and. These two functions satisfy the relation $\det (L(\lambda) - \mu \E) = 0$. 
Let us construct a map $\Phi \colon \pazocal S_{C}^K\to\Pic(X, \Sigma')$. Let $L \in \pazocal S_{C}^K$, and let $X_0 = X \setminus (\pi^{-1}(K) \cup X_\infty )$.
Take $P \in X_0$. Then the matrix $L(\lambda(P)) - \mu(P) \E$ is finite and has one-dimensional kernel. In this way, we obtain a holomorphic mapping
$
\psi \colon X_0 \to \CP^{n-1}$
which maps $P$ to  $\Ker L(\lambda(P)) - \mu(P) \E$. As it is easy to see, the mapping $\psi$ can be uniquely extended to a holomorphic mapping defined on the whole $X$ (see e.g. \cite{LaxStab}, Proposition 8.2). Obviously, we have
$\psi(P) \in \Ker( L(\lambda(P)) - \mu(P) \E)$ for each $P \in X\setminus X_\infty$, and $\psi(\infty_i)^j = \delta_{i}^{j}$ where $ \delta_{i}^{j}$ is the Kronecker delta.\par
Let $\alpha = (\alpha_1, \dots, \alpha_n) \in (\Complex^*)^n$, and let
$$
D_\alpha = \left({\sum_{i=1}^n \alpha_i \psi^i}\right)_0.
$$
In other words, we define $$h_\alpha = \psi \left({\sum \alpha_i \psi^i}\right)^{-1},$$ and set $D_\alpha = (h_\alpha)_\infty$.\par
We would like to fix $\alpha$, and set $\Phi(L) = [D_\alpha]_{\Sigma'}$ for each $L \in \pazocal S_{C}^K$. However, this is not possible, since the divisor $D_\alpha$ is not necessarily $\Sigma'$-regular. Nevertheless, for each $L \in \pazocal S_{C}^K$ we can find $\alpha$ such that $D_\alpha$ is $\Sigma'$-regular. The problem is that if we take distinct $\alpha, \beta  \in (\Complex^*)^n$, then $D_\alpha$ and $D_\beta$ are $\Sigma$-linearly equivalent, but not $\Sigma'$-linearly equivalent. Let us show how to overcome this difficulty.\par
Let $\alpha \in  (\Complex^*)^n$, and choose $f_\alpha \in \pazocal M^*(X,\Sigma)$ such that $f_\alpha(\infty_i) = \alpha_i$. Then $(f_\alpha)$ is a $\Sigma'$-regular divisor, and its $\Sigma'$-linear equivalence class does not depend on the choice of $f_\alpha$.
\begin{statement}
Assume that $D_\alpha$ and $D_\beta$ are $\Sigma'$-regular. Then $D_\alpha - (f_\alpha)\stackrel{\Sigma'}{\sim} D_\beta - (f_\beta)$.
\end{statement}
\begin{proof} Let
$$
f = \left({\sum \alpha_i \psi^i}\right)f_\alpha^{-1}\left(\sum \beta_i \psi^i\right)^{-1}{f_\beta}.$$ Then $D_\alpha - (f_\alpha) - D_\beta + (f_{\beta}) = (f)$, and
 $f \in \pazocal M^*(X,\Sigma')$.
\end{proof}
Let $U_\alpha = \{ L \in \pazocal S_{C}^K : D_\alpha \in \Div(X,\Sigma')\}$. Define $\Phi_\alpha \colon U_\alpha \to \Pic(X,\Sigma')$ by setting $$\Phi_\alpha(L) = [D_\alpha - (f_\alpha)]_{\Sigma'}.$$ Then $\pazocal S_{C}^K = \bigcup_\alpha U_\alpha$, and for each $L \in U_\alpha \cap U_\beta$, we have $\Phi_\alpha(L) = \Phi_\beta(L)$. In this way, we obtain a mapping $\Phi \colon \pazocal S_{C}^K \to  \Pic(X,\Sigma')$ given by $\Phi(L) =\Phi_\alpha( L)$ for each $L \in U_\alpha$. %As we will see in the next section, the degree of $D_\alpha$ is the same for all $L(\lambda)$
%It is clear that the mapping $\Phi$ is at least continuous.
\subsubsection{Multidegree count}\label{degCount}
We have $\deg \Phi(L) = \deg D_\alpha$. To find the total degree of $ D_\alpha$, we use a standard trick (see \cite{babelon}, Chapter 5.2). Let $a \in \CPP$ be a regular value of the function $\lambda$, and let $\lambda^{-1}(a) = \{ P_1, \dots, P_n \}$. Set $$r(a) ={ \det}^2(h_\alpha(P_1), \dots, h_\alpha(P_n)).$$ Then it is easy to see that $r$ can be extended to a meromorphic function on the whole $\CPP$. Clearly,  we have $$\deg (r)_\infty = 2|\deg (h_\alpha)_\infty| = 2|\deg D_\alpha|.$$ To count the poles of $r$, we count its zeros. 
As it is easy to see, it is possible to choose such \nolinebreak $\alpha$ that $L(\lambda) \in U_\alpha$, and the divisor $\lambda_*(D_\alpha)$ does not intersect the divisor $\lambda_*((\diffFXi{\chi}{\mu})_0)$ where $\chi$ is the defining polynomial of $C$.
Further, for the sake of simplicity, we shall assume that the curve $C$ satisfies the following genericity assumption:
for each $a \in \Complex$, we have \begin{align}\label{genas}|  C \cap \{\lambda = a\}| \geq n-1.\end{align}This means that each line $\lambda = a$ can contain either at most one node, or at most one simple ramification point. Furthermore, for all nodes of $C$, both tangents are non-vertical. It follows that the divisor $(\lambda)_a = (\lambda - a)_0$ can be of one of the following types:
\begin{enumerate}
\begin{multicols}{2}
\item $(\lambda)_a = \sum_{i=1}^n P_i$. \item  $(\lambda)_a = 2P_{n-1} + \sum_{i=1}^{n-2} P_i$. \item  $(\lambda)_a = P^+ + P^- + \sum_{i=1}^{n-2} P_i$, $\pi(P^\pm) \in K$.
 \item  $(\lambda)_a = P^+ + P^- + \sum_{i=1}^{n-2} P_i$, $\pi(P^\pm) \in \hat K$.
 \end{multicols}
\end{enumerate}
In all cases, the points $P_1, \dots, P_{n-1}, P^+, P^-$ are pairwise distinct, and $\pi(P_i) \notin \Sing(C)$.\par
\begin{statement}\label{loc0}Let $(\lambda)_a = \sum_{i=1}^n P_i$. Then $r(a) \neq 0$. \end{statement}\begin{proof}Obvious. \end{proof}\begin{statement}\label{loc1}
Let $(\lambda)_a = 2P_{n-1} + \sum_{i=1}^{n-2} P_i$. Then \begin{longenum} \item the matrix $L(a)$ has a $2\times 2$ Jordan block with eigenvalue $\mu(P_{n-1})$, eigenvector $h_\alpha(P_{n-1})$, and generalized eigenvector $( h_\alpha)'_z (P_{n-1})$ where $z$ is any local coordinate near $P_{n-1}$;
\item $r$ has a simple zero at $a$. \end{longenum}
\end{statement}
\begin{proof}
Take a local coordinate $z$ such that $\lambda = a + z^2$. Differentiating the equation $$(L(\lambda(z)) - \mu(z)E)h_\alpha(z) = 0$$ with respect to $z$, we prove item a). To prove item b), let $b \to a$, then $$r(b) =4(\lambda - a)\left({ \det}^2\left(h_\alpha(P_1), \dots, h_\alpha(P_{n-2}), h_\alpha(P_{n-1}), \diffFXYp{ h_\alpha }{z}(P_{n-1})\right) + o(1) \right).$$
\end{proof}
\begin{statement}\label{loc2}
Let $(\lambda)_a = P^+ + P^- + \sum_{i=1}^{n-2} P_i$ where $\pi(P^\pm) \in K$. Then \begin{longenum} \item the vectors $h_\alpha(P^+)$ and $h_\alpha(P^-)$ are linearly independent; \item $r(a) \neq 0$. \end{longenum}
\end{statement}
\begin{proof}
If we assume that $h_\alpha(P^+)$ and $h_\alpha(P^-)$ are linearly dependent, then we necessarily have $h_\alpha(P^+) = h_\alpha(P^-)$. Take \nolinebreak$\lambda$ as a local coordinate near $P^+$ and $P^-$. Differentiating $$(L(\lambda) - \mu E)h_\alpha(\lambda) = 0$$ with respect to $\lambda$ at $P^+$ and $P^-$ and subtracting the obtained equations, we get
\begin{align}\label{nodeEq}
(L(a) - \mu(P^\pm) E)\left(\diffFXYp{ h_\alpha}{ \lambda}(P_{+})  -\diffFXYp{ h_\alpha}{ \lambda}(P_{-})\right) = \left(\diffFXYp{ \mu}{ \lambda}(P_{+})  -\diffFXYp{ \mu}{ \lambda}(P_{-})\right) h_\alpha(P^\pm).
\end{align}
Since the singular point $\pi(P^\pm)$ is nodal, we have $\diff \mu / \diff \lambda(P_{+})  \neq \diff \mu / \diff \lambda(P_{-})$, so $L_a$ has a Jordan block, which contradicts the definition of the set $K$. This proves item a). Item b) obviously follows.
\end{proof}
\begin{statement}\label{loc3}
Let $(\lambda)_a = P^+ + P^- + \sum_{i=1}^{n-2} P_i$ where $\pi(P^\pm) \in \hat K$. Then  \begin{longenum}  \item the matrix $L(a)$ has a $2\times 2$ Jordan block with eigenvalue $\mu(P^{+}) = \mu(P^-)$, eigenvector $h_\alpha(P^{+}) = h_\alpha(P^{-} )$, and generalized eigenvector $ (h_\alpha)'_\lambda(P_{+})  -  (h_\alpha)'_\lambda(P_{-})$;
\item $r$ has a double zero at $a$.  \end{longenum}
\end{statement}
\begin{proof}
Since the geometric multiplicity of $\mu(P^{\pm}) $ is equal to $1$, we have $h_\alpha(P^+) = h_\alpha(P^-)$, which implies equation \eqref{nodeEq} and hence item a).  To prove item b), let $b \to a$, then $$r(b) =(\lambda - a)^2\left({ \det}^2\left(h_\alpha(P_1), \dots, h_\alpha(P_{n-2}), h_\alpha(P^+), \diffFXYp{ h_\alpha}{ \lambda}(P_{+})  -\diffFXYp{ h_\alpha}{ \lambda}(P_{-})\right) + o(1) \right).$$
\end{proof}

Let $$D_\Sigma = \sum\nolimits_{i=1}^{|\Sigma|} (P_i^+ + P_i^-).$$ Considering Propositions \ref{loc1} - \ref{loc3}, we conclude with the following:
%%Without loss of generality, we may assume that $\pi(\mathrm{supp (D_\alpha)})$ does not intersect $\Sing C$, and 
\begin{statement}
We have $(r)_0 = \lambda_*((\lambda)_R + D_\Sigma)$.
\end{statement}
%In Case a), we obviously have $r(a) \neq 0$. \par
Using \eqref{geomGenus}, we conclude that
$$\deg D_\alpha = \frac{1}{2}\deg (r)_\infty =  \frac{1}{2}\deg (r)_0 = \frac{mn(n-1)}{2} - |K|= g(X,\Sigma').$$\par
Now, let us prove that $\deg D_\alpha$ is uniform. Let $X_I /\Sigma'_I$ be a subcurve of $X / \Sigma'$, and let $q = \deg \lambda\mid_{X_I}$. Then $|X_I \cap X_\infty| = q$. Without loss of generality, we may assume that $\infty_1, \dots, \infty_q \in X_I$. Let $pr \colon \Complex^n \to \Complex^q$ be a map given by $pr(x^1, \dots, x^n) = (x^1, \dots, x^q)$. Let also $a \in \CPP$ be a regular value of the function $\lambda$. Then $\lambda^{-1}(a) \cap X_I$ consists of $q$ points $ P_1, \dots, P_q$. Set $$r_I(a) = {\det}^2(pr(h_\alpha(P_1)), \dots, pr(h_\alpha(P_q))).$$ We have $r_I(\infty) \neq 0$, so $r_I \not\equiv 0$. Repeating the above arguments, we get  $$  |\deg D_\alpha\mid_{X_I}| \geq \frac{1}{2}\deg (r_I)_\infty =  \frac{1}{2}\deg (r_I)_0 \geq g(X_I, \Sigma'_I),$$ q.e.d.

\subsubsection{Injectivity}
\begin{statement}\label{nonSpec1}
Let $L \in U_\alpha$. Then $\dim \mathrm{L}(D_\alpha -D_\infty, \Sigma) = 0$.
\end{statement}
\begin{proof}
Let $ \hat \lambda \colon \mathrm{L}(D_\alpha -D_\infty, \Sigma) \to \mathrm{L}(D_\alpha, \Sigma) $ and $ A \colon \mathrm{L}(D_\alpha, \Sigma) \to \mathrm{L}(D_\alpha, \Sigma -D_\infty) $ be linear maps given by $$ \hat \lambda f = \lambda f, \quad A(f) = f- \sum_{i=1}^n \alpha_if(\infty_i)h_\alpha^i.$$ Assume that $\dim \mathrm{L}(D_\alpha -D_\infty, \Sigma) > 0$. Then the operator $$A\hat \lambda \colon  \mathrm{L}(D_\alpha -D_\infty, \Sigma) \to  \mathrm{L}(D_\alpha -D_\infty, \Sigma) $$ must have an eigenvector $g_\alpha \in  \mathrm{L}(D_\alpha -D_\infty, \Sigma) $. Denote the corresponding eigenvalue by $a$. We have \begin{align}\label{badEquation}(\lambda - a) g_\alpha = \sum_{i=1}^n c_i h_\alpha^i\end{align} where $c_1, \dots, c_n \in \Complex$. As it is easy to see, the spectrum of $p\hat \lambda$ does not depend on the choice of $\alpha$, so we can assume that $\mathrm{supp}(D_\alpha) \cap \lambda^{-1}(a)$ is empty. Then \eqref{badEquation} implies that
\begin{align}\label{bad2} \sum_{i=1}^n c_i h_\alpha^i(P) = 0 \,\, \forall \,\, P \in \lambda^{-1}(a).\end{align}
Let us assume that genericity assumption \eqref{genas} is satisfied, and hence $(\lambda)_a$ belongs to one of the four aforementioned types. Let us consider each of these types and show that \eqref{badEquation} can not hold. The proof in the general case is analogous.
\begin{enumerate}\item Let $(\lambda)_a = \sum_{i=1}^n P_i$. Then \eqref{bad2} implies that $h_\alpha(P_1), \,\dots$, $h_\alpha(P_n)$ are linearly dependent, which is not possible.
\item Let $(\lambda)_a = 2P_{n-1} + \sum_{i=1}^{n-2} P_i$. Let $z$ be a local coordinate near $P_{n-1}$. We have $\lambda' _z(P_{n-1}) = 0$, so by \eqref{badEquation} we have $$ \sum_{i=1}^n c_i  \diffFXYp{h_\alpha^i}{z}(P_{n-1}) = 0$$ and, using \eqref{bad2}, we conclude that $h_\alpha(P_1), \,\dots$, $h_\alpha(P_{n-2})$, $ h_\alpha(P_{n-1})$, $(h_\alpha)'_{z}(P_{n-1})$ are linearly dependent. By item a) of Proposition \ref{loc1}, this is not possible.
\item Let $(\lambda)_a = P^+ + P^- + \sum_{i=1}^{n-2} P_i$ where $\pi(P^\pm) \in K$. In view of Proposition \ref{loc2}, this case is analogous to Case 1.
\item Let $(\lambda)_a = P^+ + P^- + \sum_{i=1}^{n-2} P_i$ where $\pi(P^\pm) \in \hat K$.  By \eqref{badEquation}, we have
\begin{align}\label{bad3}
%g_\alpha(P^\pm) = \sum c_i \diffFX{ h_\alpha^i }{ \lambda }(P^{\pm}) \Longrightarrow
 g_\alpha(P^+) - g_\alpha(P^-) =  \sum_{i=1}^n c_i \left(\diffFXYp{ h_\alpha^i }{ \lambda }(P^{+}) -  \diffFXYp{ h_\alpha^i }{ \lambda }(P^{-}) \right).
\end{align}
Since $\g_\alpha \in \pazocal M(X, \Sigma)$, we should have $ g_\alpha(P^+) = g_\alpha(P^-) $, so \eqref{bad2} and \eqref{bad3} imply that 
$
h_\alpha(P_1),\, \dots$, $h_\alpha(P_{n-2})$, $h_\alpha(P^+)$, $(h_\alpha)'_\lambda(P_{+})  - (h_\alpha)'_\lambda(P_{-})$
are linearly dependent. This is impossible by item a) of Proposition \ref{loc3}.
\end{enumerate}

\end{proof}
\begin{statement}\label{nonSpec2}
Let $L \in U_\alpha$. Then $\dim \mathrm{L}(D_\alpha -D_\infty + \infty_i, \Sigma) = 1$.
\end{statement}
\begin{proof}
Let $f,g \in  \mathrm{L}(D_\alpha -D_\infty + \infty_i, \Sigma) $. Then $f(\infty_i)g - g(\infty_i)f \in \mathrm{L}(D_\alpha -D_\infty, \Sigma) $. So, by Proposition \ref{nonSpec1}, we have $f(\infty_i)g - g(\infty_i)f = 0$, and $\dim \mathrm{L}(D_\alpha -D_\infty + \infty_i, \Sigma) \leq 1$. On the other hand, $h_\alpha^i \in \mathrm{L}(D_\alpha -D_\infty + \infty_i, \Sigma)$, so $\dim \mathrm{L}(D_\alpha -D_\infty + \infty_i, \Sigma) = 1$.
\end{proof}
\begin{statement}\label{nonSpec3}
Let $L \in U_\alpha$. Then $\dim \mathrm{L}(D_\alpha, \Sigma') = 1$.
\end{statement}
\begin{proof}
 Consider the linear map
$
A \colon \mathrm{L}(D_\alpha, \Sigma') \to \mathrm{L}(D_\alpha - D_\infty, \Sigma)
$
given by $A(f) = f - f(\infty_1)$. We have $$\dim \mathrm{L}(D_\alpha , \Sigma') \leq \dim \mathrm{L}(D_\alpha - D_\infty, \Sigma)  + \dim \Ker A = 1.$$ On the other hand, we have $1 \in \mathrm{L}(D_\alpha, \Sigma')$, so  $\dim \mathrm{L}(D_\alpha, \Sigma') = 1$.
\end{proof}
Now, let us proof that $\Phi$ is injective. Assume that $L^{(1)} \neq L^{(2)}$, and that $\Phi(L^{(1)}) =\Phi( L^{(2)})$. As it is easy to see, there exists $\alpha$ such that $L^{(1)}, L^{(2)} \in U_\alpha$. We have $$[D^{(1)}_\alpha]_{\Sigma'} =[D^{(2)}_\alpha]_{\Sigma'} , $$ so $D^{(2)}_\alpha - D^{(1)}_\alpha = (f)$ where $f \in \mathrm{L}(D^{(1)}_\alpha, \Sigma')$. By Proposition \ref{nonSpec3}, we have $f = \const$, therefore $D^{(2)}_\alpha = D^{(1)}_\alpha$.\par
Further, let us show that  $L^{(1)} = L^{(2)}$. We have $$(h^{(1)}_\alpha)^i, (h^{(2)}_\alpha)^i \in \mathrm{L}(D^{(1)}_\alpha -D_\infty + \infty_i, \Sigma),$$ and using Proposition \ref{nonSpec2}, we conclude that $(h^{(1)}_\alpha)^i$ and $(h^{(2)}_\alpha)^i$ are proportional. At the same time,  we have $(h^{(1)}_\alpha)^i(\infty_i)=(h^{(2)}_\alpha)^i(\infty_i) = (\alpha_i)^{-1}$, so $(h^{(1)}_\alpha)^i=(h^{(2)}_\alpha)^i$, and $h^{(1)}_\alpha=h^{(2)}_\alpha$. Consequently, for each $a \in \Complex$, the matrices $L^{(1)}(a)$ and $L^{(2)}(a)$ have same eigenvalues and eigenvectors, and must coincide.

\subsubsection{Denseness of the image}
\begin{statement}\label{uss} Let $d$ be a multidegree of total degree $g(X,\Sigma')$. Then $d$ is uniform on $X/\Sigma'$ if and only if $d - \deg D_\infty$ is semistable on $X/\Sigma$.
\end{statement}
\begin{proof}
This follows from the obvious formula
$$
g(X_I, \Sigma'_I) -  \left\lvert\deg D_\infty \mid_{X_I}\right\rvert= g(X_I, \Sigma_I) - c(X_I,\Sigma_I).
$$
satisfied for any $I \subset \left\{1,\dots, c(X)\right\}$.
\end{proof}
Let $d$ be a uniform degree on  $X/\Sigma'$. By Proposition \ref{density}, the set $E_d(X,\Sigma')$ is dense in $\Pic_d(X,\Sigma')$. Further, let $d_r = d - \deg D_\infty$. Then, by Proposition \ref{semistabMain}, the set
$W_{d_r}(X,\Sigma)$ has positive codimension in $\Pic_{d_r}(X,\Sigma)$. Let
$$
\Pic^{reg}_d(X,\Sigma') = E_d(X,\Sigma) \cap (i^*_{\Sigma})^{-1}\left(\Pic_{d_r}(X,\Sigma) \setminus W_{d_r}(X,\Sigma) + [D_\infty]\right).
$$
The set $\Pic^{reg}_d(X,\Sigma')$ is dense in $\Pic_d(X,\Sigma')$. Let us show that $\Imm \Phi \supset \Pic^{reg}_d(X,\Sigma')$, so that $\Imm \Phi$ is also dense. Let $\xi \in \Pic^{reg}_d(X,\Sigma')$. Then we can find a $\Sigma'$-regular effective divisor $D$ such that $[D] = \xi$. 
%Furthermore, we have $\dim \mathrm L(D-D_\infty, \Sigma) = 0$. 
By Riemann's inequality, we have $$\dim \mathrm L(D - D_\infty + \infty_i, \Sigma) \geq 1.$$ Let $h^i \in L(D - D_\infty + \infty_i, \Sigma) \setminus \{0\}$. By the construction of the set $\Pic^{reg}_d(X,\Sigma')$, we have $$\dim \mathrm L(D-D_\infty, \Sigma) = 0,$$ so $h^i(\infty_i) \neq 0,$ and we can normalize $h^i$ by $h^i(\infty_i) = 1$. Define $h = (h^1, \dots, h^n)$. \pagebreak[3] We need to show that there exists $L \in \pazocal S_{C}^K$ such that 
\begin{align}\label{whatweneed}
(L(\lambda(P)) -\mu(P)\E)h(P) = 0.
\end{align}
%We have $$\left(\sum_{i=1}^n h^i \right)- 1 \in  \mathrm L(D-D_\infty, \Sigma), $$ so 
%$\sum_{i=1}^n h^i =1.
%$ 
\par
Let $a \in \CPP$ be a regular value of $\lambda$, and let $\lambda^{-1}(a) = \{ P_1, \dots, P_n \}$. Let
$$r(a) ={ \det}^2(h(P_1), \dots, h(P_n)).$$
\begin{statement}\label{loc4}
Proposition \ref{loc0}, item b) of Proposition \ref{loc1}, item b) of Proposition \ref{loc2}, and item b) of Proposition \ref{loc3} hold for $r(a)$.
\end{statement}
\begin{proof}
Arguments similar to that of Section \ref{degCount} show that $(r)_0 \geq \lambda_*((\lambda)_R + D_\Sigma)$, and that if at least one of Propositions \ref{loc0} -- \ref{loc3} does not hold, then this inequality must be strict. Comparing degrees we conclude that 
$(r)_0 = \lambda_*((\lambda)_R + D_\Sigma)$, which proves the proposition. %Let us consider the matrix $H(a) = (h^j(P_i))$,
\end{proof}
%and let $r(a) = \det H(a)^2$.
Define matrices $H(a) = (h^j(P_i))$ and $M(a) = \mathrm{diag}(\mu(P_1), \dots , \mu(P_n))$. Let
$$
L(a) = H(a)M(a)H(a)^{-1}.
$$
Then $L(a)$ is meromorphic in $a$ and satisfies \eqref{whatweneed}. 
Local analysis using Proposition \ref{loc4} shows that $L(a)$ does not have poles except for the pole at infinity, and that $L(a) -b\E$ has two-dimensional kernel if and only if $(a,b) \in K$.  Finally,  the condition $C \in \mathcal C_{spec}$ implies that the pole at infinity is of order $m$ and that the leading term is equal to $J$. Therefore, we have $L \in \pazocal S_{C}^K$, and $\Phi(L) = \xi$, q.e.d.

\subsubsection{Linearization of flows}
Let us consider the solution curve of \eqref{loopI} and show that its image under the mapping $\Phi$ is given by \eqref{velocity}. The proof is similar to the non-singular case. Denote
$
A(\lambda) = \phi(L(\lambda),\lambda^{-1})_+.
$
In a standard way, we show that $h_\alpha$ satisfies
\begin{align}\label{heq}
\diffXp{t} h_\alpha = (\nu\E - A(\lambda) )h_\alpha
\end{align}
where 
$$
\nu(P) = \sum_{i=1}^n \alpha_i (A(\lambda(P))h_\alpha(P))^i \in \pazocal M(X,\Sigma).
$$
%is a $\Sigma$-regular function on $X$.
Let $P_1(t), \dots, P_g(t)$ be the poles of $h_\alpha$, and let $D_\alpha(t) = \sum P_i(t)$. Then \eqref{heq} implies that
\begin{align*}
\diffXp{t} \int_{D_\alpha(0)}^{D_\alpha(t)}\omega=-\sum_{i=1}^g \Res_{P_i}\,\nu\omega%,\quad \xi \in \Pic_{g'}(X(p)),\quad \omega \in \Omega^1(X(p))
\end{align*}
where $\omega$ is any meromorphic differential. \par \pagebreak[3]
Now, let $\omega \in \Omega^1(X,\Sigma')$. Then %Since $\lambda, h_\alpha^1, \dots, h_\alpha^n \in \pazocal M(X, \Sigma)$, we have $\omega_q \in \Omega^1(X,\Sigma)$, so
\begin{align*}
-\sum_{i=1}^g \Res_{P_i}\,\nu\omega = \sum_{i=1}^n \Res_{\infty_i}\, \nu\omega + \sum_{i=1}^{|\Sigma|}\left(\Res_{P_i^+}\, \nu\omega + \Res_{P_i^-}\, \nu\omega\right). % -\sum \sum_{j=1}^g \Res_{P_j}\left(  \sum_{i=1}^n \alpha_i(A(\lambda)h_\alpha)^i\right)\omega
%,\quad \xi \in \Pic_{g'}(X(p)),\quad \omega \in \Omega^1(X(p))
\end{align*}
Since  $\omega \in \Omega^1(X,\Sigma')$, and $\nu \in \pazocal M(X,\Sigma)$, the latter sum vanishes. At the same time, we have
\begin{align*}
 \Res_{\infty_i}\, \nu\omega = \sum_{j=1}^n \alpha_j \, \Res_{\infty_i}\, \left( \phi(L(\lambda),\lambda^{-1})_+\,h_\alpha \right)^j \omega.%,\quad \xi \in \Pic_{g'}(X(p)),\quad \omega \in \Omega^1(X(p))
\end{align*}
Note that $\mathrm{ord}_{\infty_i} \left(  (\phi(L(\lambda),\lambda^{-1}) - \phi(L(\lambda),\lambda^{-1})_{+})h_\alpha \right)^j \geq 1$, and $\mathrm{ord}_{\infty_i} \omega \geq -1$, so
$$
\Res_{\infty_i}\, \left( \phi(L(\lambda),\lambda^{-1})_+\,h_\alpha \right)^j \omega = \Res_{\infty_i}\, \left( \phi(L(\lambda),\lambda^{-1})\,h_\alpha \right)^j \omega = \Res_{\infty_i}\,  \phi(\mu,\lambda^{-1})\,h_\alpha^j \,\omega,
$$
and
$$
\diffXp{t} \int_{D_\alpha(0)}^{D_\alpha(t)}\omega = \sum_{i=1}^n \Res_{\infty_i}\, \nu\omega =  \sum_{i=1}^n \sum_{j=1}^n \alpha_j  \,\Res_{\infty_i}\, \left( \phi(\mu,\lambda^{-1})\,h_\alpha \right)^j \omega = \sum_{i=1}^n  \,\Res_{\infty_i}\,  \phi(\mu,\lambda^{-1}) \omega, 
$$
q.e.d.
%
%Now, let us show that flows  span the tangent space to  $\Pic(X,\Sigma')$. \par
%In order to do this, we explicitly describe the space $\Omega^1(X,\Sigma')$.
%Let us define the weighted degree of a monomial $\lambda^i\mu^j$ as $mi + j$. The weighted degree $\mathrm{wdeg}\, r$of a polynomial $r(\lambda,\mu)$ is the highest weighted degree of its terms.
%\begin{statement}
%
%$\Omega^1(X,\Sigma')$ consists of differentials of the form
%
%$$
%\frac{r(\lambda,\mu) \diff \lambda}{p_\mu(\lambda, \mu)}
%$$
%where $\mathrm{wdeg}\,r \leq m(n-1)-1$, and $r(\lambda, \mu) = 0$ for each $(\lambda, \mu) \in K$.
%\end{statement}
%\begin{proof}
%As it is easy to see, all differentials of this form are $\Sigma'$-regular, and dimensions coincide.
%%$$
%%\frac{mn(n-1)}{2} - |K| = g(X,\Sigma'),
%%$$
%%which proves the proposition.
%\end{proof}
\begin{statement}\label{completeness}
Flows \eqref{velocity} span the tangent space to $\Pic(X,\Sigma')$.
\end{statement}
\begin{proof}
Let us consider a bilinear pairing
$$
\langle \,,  \rangle_\infty \colon \Complex[\mu,\lambda^{-1}] \times  \Omega^1(X,\Sigma') \to \Complex
$$
given by
\begin{align}\label{HMAP}
\langle \phi, \omega \rangle_\infty = \sum_{i=1}^n  \,\Res_{\infty_i}\,  \phi \omega.
\end{align}
\pagebreak[4]
We need to show that the mapping $ \Complex[\mu,\lambda^{-1}]  \to  \Omega^1(X,\Sigma')^*$ given by $\phi \mapsto \langle \phi, \,\rangle_\infty$ 
is surjective, or, which is the same, that the right radical of the form $\langle \,,  \rangle_\infty$ is trivial.
 Let
$
s = \max_{i} \mathrm{ord}_{\infty_i}\,\omega.
$
Then
$$
 \mathrm{ord}_{\infty_i}\,\mu^j\lambda^{-k}\omega \geq k - mj + s,
$$
and if $k - mj + s = -1$, then
$$
\Res_{\infty_i}\,\mu^j\lambda^{-k}\omega = j_i^j \,\Res_{\infty_i}\,\lambda^{s+1}\omega, 
$$
and
$$
\langle \mu^j\lambda^{-k}, \omega \rangle_\infty =  \sum_{i=1}^n  j_i^j \,\Res_{\infty_i}\,\lambda^{s+1}\omega.
$$
Assume that
\begin{align}\label{resFormula}
 \sum_{i=1}^n  j_i^j \,\Res_{\infty_i}\,\lambda^{s+1}\omega = 0 \,\, \forall \,\, j,k \geq 0 : k - mj + s = -1.
\end{align}
Consider two cases.
\begin{enumerate}
\item
If $j_i \neq 0$ for each value of $i$, then \eqref{resFormula} implies that $\Res_{\infty_i}\,\lambda^{s+1}\omega = 0$ for each $i$, which contradicts the choice of $s$.
%Therefore, there exist $j,k \geq 0$ such that $\langle \mu^j\lambda^{-k}, \omega \rangle \neq 0$, q.e.d. 
\item If, say, $j_1 = 0$, then \eqref{resFormula} implies that $\Res_{\infty_i}\,\lambda^{s+1}\omega = 0$ for $2 \leq i \leq n$. Therefore, according to our choice of $s$, we have $\Res_{\infty_1}\,\lambda^{s+1}\omega \neq 0$. At the same time, since $s+1 \geq 0$,  the differential $\lambda^{s+1}\omega$ may have poles only at $\infty_1, \dots, \infty_n$ and points of $\mathrm{supp}(\Sigma)$, and
$$
\mathrm {res}_{P_i^+} \,\lambda^{s+1}\omega + \mathrm {res}_{P_i^-} \,\lambda^{s+1}\omega = 0,
$$
therefore
$$
  \sum_{i=1}^n \Res_{\infty_i}\,\lambda^{s+1}\omega = 0.
$$

\end{enumerate}
So we have a contradiction in both cases, which proves that $\langle \mu^j\lambda^{-k}, \omega \rangle_\infty \neq 0$ for some non-negative $j,k $, q.e.d.
\end{proof}

\subsubsection{Smoothness of $\pazocal S_{C}^K$}
Among the flows \eqref{loopI}, there is a finite number of linearly independent, say, $N$. These flows generate a local $\Complex^N$ action on $\pazocal S_{C}^K$. 
Let $L \in \pazocal S_{C}^K$, and let $O(L)$ be its local orbit under the  $\Complex^N$ action. By Proposition \ref{completeness}, there exists a neighborhood of $\Phi(L)$ which is completely contained in $\Phi(O(L))$. Since the map $\Phi$ is continuous and injective, this implies that there exists a neighborhood $U(L) $ such that $$U(L) \cap \pazocal S_{C}^K = U(L) \cap O(L),$$
therefore $\pazocal S_{C}^K$ is a complex analytic manifold. The map $\Phi $ is bijective and linear in a coordinate chart induced by the $\Complex^N$ action, so it is biholomorphic. Further, Proposition \ref{completeness} implies that flows \eqref{loopI} span the tangent space to $\pazocal S_{C}^K$, q.e.d.
\subsection{Argument shift method and integer points in permutohedra}\label{ex1}
When $m=1$, the space $\mathcal L_m^J(\gl(n)) = \{ X+\lambda J \mid X \in \gl(n)\}$ can be naturally identified with $\gl(n)$. In this case, the integrable system $\pazocal F$ coincides with the system constructed by the so-called \textit{argument shift method}\footnote{Note that if we restrict this system to the subspace $L(\lambda)^t = -L(-\lambda)$, then for a certain choice of $\phi$ in \nolinebreak \eqref{loopI}, we obtain the equation of the free $n$-dimensional rigid body \cite{Manakov}.} \cite{MF}. %This system or, more precisely, its restriction to the subspace $X^t = -X, J^t = J$ was first studied in \cite{Manakov}.\par
Let us assume that $J = \mathrm{diag}(j_1, \dots, j_n)$ and consider a curve $C$ given by
\begin{align}\label{degCurve}
\prod_{i=1}^n\,( \alpha_i + \lambda j_i - \mu) = 0.
\end{align}
where $\alpha_1, \dots, \alpha_n \in \Complex$. We assume that the curve \eqref{degCurve} is nodal which is equivalent to the condition that the lines $l_1, \dots, l_n$ where $l_i = \{(\lambda, \mu) \in \Complex^2 \mid\alpha_i + \lambda j_i - \mu = 0\}$ are in general position. \par
It is clear that the level set $\pazocal S_C$ contains at least a point $L = \mathrm{diag}(\alpha_1, \dots, \alpha_n)$ which is a common fixed point for all flows \eqref{loopI}, i.e. it is a rank $0$ point for $\F$ (see Section \ref{nsis}). Further, let $\succ$ be any ordering on the set $\{1, \dots, n\}$. Consider the Borel subalgebra 
$$
\mathfrak b_\succ = \{ L \in \gl(n) \mid L_{ij} = 0 \,\,\forall \,\, i \succ j \}
$$
and the corresponding maximal nilpotent subalgebra
$$
\mathfrak n_\succ =[\mathfrak b_\succ, \mathfrak b_\succ]= \{ L \in \mathfrak b_\succ \mid L_{ii} = 0 \}.
$$
We note that subalgebras $b_\succ$ are exactly those Borel subalgebras which contain the centralizer of $J$. There are $n!$ of them, corresponding to the number of elements in the Weil group of $\gl(n)$.\par
Let $\mathfrak q_\succ$ be the coset $$\mathfrak q_\succ =  \mathrm{diag}(\alpha_1, \dots, \alpha_n) + \mathfrak n_\succ \subset \mathfrak b_\succ.$$
Then we have $\mathfrak q_\succ \subset \pazocal S_C$. 
Comparing dimensions, we conclude that $\mathfrak q_\succ $ has an open subset $\mathfrak q_\succ^0 $ completely contained in the regular part  $\pazocal S_C^\emptyset \subset  \pazocal S_C$, so $\pazocal S_C^\emptyset $ has at least $n!$ connected components, and $\pazocal S_C $ has at least $n!$ irreducible components.

\par 
However, in fact, there are much more. By Theorem \ref{thm1}, components of $\pazocal S_C^\emptyset $ are in one-to-one correspondence with uniform multidegrees on the curve obtained from $C$ be identifying points at infinity. The set of uniform multidegrees on this curve coincides with the set of integer points in the polytope
\begin{align*}
P_n = \left\{ x \in \R^n : \sum_{i=1}^n  x_i = \frac{n(n-1)}{2}; \,\, \sum_{i \in I} x_i \geq \frac{|I|(|I|-1)}{2}  \,\, \forall \,\, I \subset \{ 1, \dots, n\} \right\}.
\end{align*}
known as \textit{permutohedron}. This polytope is the convex hull of the set of points $$V_n = \{ v_\sigma = (\sigma(0), \dots, \sigma(n-1)) \in \R^n \mid \sigma \in S_n\}.$$ 
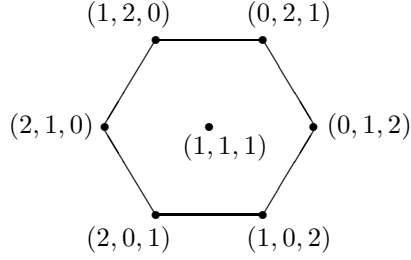
\begin{figure}[t]
\centerline{
\begin{picture}(100,100)
\put(50,50){
\begin{picture}(200,200)
\put(-20, 33){\line(1,0){40}}
\put(20, 33){\line(3,-5){20}}
\put(20, -33){\line(3,5){20}}
\put(-20, -33){\line(1,0){40}}
\put(-20, 33){\line(-3,-5){20}}
\put(-20, -33){\line(-3,5){20}}
\put(39,0){\circle*{3}}
\put(44,-2){$(0,1,2)$}
\put(14,-45){$(1,0,2)$}
\put(-75,-2){$(2,1,0)$}
\put(-46,-45){$(2,0,1)$}
\put(14,40){$(0,2,1)$}
\put(-46,40){$(1,2,0)$}
\put(-10,-10){$(1,1,1)$}
\put(-39,0){\circle*{3}}
\put(20,33){\circle*{3}}
\put(-20,33){\circle*{3}}
\put(20,-33){\circle*{3}}
\put(-20,-33){\circle*{3}}
\put(0,0){\circle*{3}}
\end{picture}
}
\end{picture}
}
\caption{Integer points in permutohedron $P_3$}\label{perm}
\end{figure}
As it is not difficult to see from the construction of the map $\Phi$ (see Section \ref{cons}), the $n!$ vertices $v_\sigma$ of the permutohedron $P_n$ correspond to components $\mathfrak q_\succ^0 $ described above. At the same time, for $n \geq 3$, there are integer points in the interior of $P_n$ as well (see Figure \ref{perm}). If $n$ is large, the number of integer points\footnote{It is also known that the number of integer points in the permutohedron $P_n$ equals the number of forests on $n$ labeled vertices \cite{postnikov2009permutohedra}.}  in $P_n$ is approximately $$\mathrm{Vol}(P_n) = n^{n-2}$$
which is much more than $n!$.\par
It is also not difficult to explicitly write down solutions of \eqref{loopI} corresponding to vertices of the permutohedron, i.e. lying in Borel subalgebras  $\mathfrak b_\succ$. For example, let $n=3$ and let $\phi = \mu^2\lambda^{-1}$. The corresponding vector field \eqref{loopI} reads
\begin{align}\label{mfe}
\dot L = [L^2, J].
\end{align}
The solutions corresponding to the vertex $(0,1,2)$ are
$$
L(t) = \left(\begin{array}{ccc}\alpha_1 & L_{12}(t) & L_{13}(t) \\0 & \alpha_2 & L_{23}(t) \\0 & 0 & \alpha_3\end{array}\right)
%L(t) = \left(\begin{array}{ccc}\alpha_1 & c_{12}\exp(\Delta_{12}\cdot t) &  c_{13}\exp(\Delta_{13}\cdot t) + \frac{c_{12}c_{13}(j_1 - j_3)}{p} \\0 & \alpha_2 &  c_{23}\exp(\Delta_{23}\cdot t)  \\0 & 0 & \alpha_3\end{array}\right)
%, \mbox{and} \quad L(t) = \left(\begin{array}{ccc}\alpha_1 & 0 &  0 \\c_{21}\exp(\Delta_{21}\cdot t) & \alpha_2 &  0  \\c_{31}\exp(\Delta_{31}\cdot t)  & c_{32}\exp(\Delta_{32}\cdot t)  & \alpha_3\end{array}\right)
$$ 
where 
\begin{align}\label{flas}
\begin{aligned}
\vphantom{L_{23}(t)\sigma^{-1}}&\quad\quad\quad L_{12}(t) = c_{12}e^{\sigma_{12}t},\,\,
L_{23}(t) = c_{23}e^{\sigma_{23}t},\quad \\
&L_{13}(t) = c_{13}e^{\sigma_{13}t} +c_{12}c_{13} {(j_1 - j_3)}{\sigma^{-1}}e^{(\sigma_{12}+\sigma_{23})t},\\
\vphantom{L_{23}(t)\sigma^{-1}}\sigma_{12} = (j_2 - j_1)(\alpha_1 &+ \alpha_2),  \,\,\sigma_{23} = (j_3 - j_2)(\alpha_2 + \alpha_3),  \,\,\sigma_{13} = (j_3 - j_1)(\alpha_1 + \alpha_3),\\
\vphantom{L_{23}(t)\sigma^{-1}}
&\sigma = \alpha_1(j_3 - j_2) + \alpha_2(j_1 - j_3) +  \alpha_3(j_2 - j_1),
\end{aligned}
\end{align}
and $c_{12}, c_{23}, c_{13} \in \Complex^*$ are arbitrary non-zero constants (if they are zero, we obtain solutions not belonging to $\pazocal S_C^\emptyset$). %We note that if we are in the real setting, and the numbers $\sigma_{12}, \sigma_{13},\sigma_{23}$ are all negative, then
%$L(t) \to  \mathrm{diag}(\alpha_1, \dots, \alpha_n) $ when $t \to +\infty$. If  $\sigma_{12}, \sigma_{13},\sigma_{23}$ are not all negative, then this property holds for some other vertex of $P_n$.

\par\pagebreak[4]

In general, all solutions of \eqref{loopI} corresponding to vertices of $P_n$ are linear combinations of exponents. In particular, they are entire functions, which means that the set $\Upsilon_d$ is empty for each $d \in V_n$, and the union $\bigsqcup_{d \in V_n}\Pic_d(X,\Sigma)$ is completely contained in the image of the map $\Phi$ . For points in the interior of $P_n$, this is no longer so. Let us again consider the case $n=3$.  The only integer point in the interior of $P_3$ is $(1,1,1)$ (see Figure \ref{perm}). The corresponding solution of \eqref{mfe} reads:
$$
L(t) = \left(\begin{array}{ccc}\alpha_1 & L_{12}^+(t) & L^-_{13}(t) \\L^-_{21}(t) & \alpha_2 & L^+_{23}(t) \\L^+_{31}(t) & L^-_{32}(t) & \alpha_3\end{array}\right)
%L(t) = \left(\begin{array}{ccc}\alpha_1 & c_{12}\exp(\Delta_{12}\cdot t) &  c_{13}\exp(\Delta_{13}\cdot t) + \frac{c_{12}c_{13}(j_1 - j_3)}{p} \\0 & \alpha_2 &  c_{23}\exp(\Delta_{23}\cdot t)  \\0 & 0 & \alpha_3\end{array}\right)
%, \mbox{and} \quad L(t) = \left(\begin{array}{ccc}\alpha_1 & 0 &  0 \\c_{21}\exp(\Delta_{21}\cdot t) & \alpha_2 &  0  \\c_{31}\exp(\Delta_{31}\cdot t)  & c_{32}\exp(\Delta_{32}\cdot t)  & \alpha_3\end{array}\right)
$$ 
where
\begin{align*}
L_{ij}^+(t) = \frac{ c_{ij}e^{\sigma_{ij}t}}{1 - \rho e^{-\sigma t}}, \quad L_{ij}^-(t) = \frac{ c_{ij}e^{\sigma_{ij}t}}{1 - \rho^{-1} e^{\sigma t}},
\end{align*}
$\sigma_{12}, \sigma_{13},\sigma_{23},\sigma$ are the same as in \eqref{flas}, $\sigma_{ij} = -\sigma_{ji}$, and the constants $c_{ij}, \rho$ satisfy
$$
\frac{c_{12}c_{21}}{j_2 - j_1} = \frac{c_{23}c_{32}}{j_3 - j_2} = \frac{c_{31}c_{13}}{j_1 - j_3} =  -\frac{c_{12}c_{23}c_{31}}{\rho}
 = \frac{\sigma^2}{(j_2-j_1)(j_3-j_2)(j_1 - j_3)}. $$
 More generally, it can be seen from the constructions of the present paper that solutions of \eqref{loopI} corresponding to all integer points in $P_n$ for arbitrary $n$ are rational functions of exponents. Apparently, there should be some combinatorics relating the permutohedron and these rational functions.
\section{Nodal curves and non-degenerate singularities of integrable systems}\label{sec2}

\subsection{Non-degenerate singularities of integrable systems}\label{nsis}

Let $(M^{2n}, \omega)$ be a real analytic or complex analytic symplectic manifold. Let us denote the space of analytic functions on $M^{2n}$ by $\pazocal O(M)$.
The space $\pazocal O(M)$ is a Lie algebra with respect to the Poisson bracket.

\begin{definition}
%A commutative subalgebra $\pazocal F \subset \Cont^\infty(M^{2n})$ is called \textit{complete at a point} $x \in M^{2n}$ if  $\dim \diff \F(x) = n$, where
%	 $\diff \F(x) = \{ \diff f(x), f \in \F\} \subset \T^*_x M$. \par 
	 Let $\pazocal F \subset \pazocal O(M)$ be a Poisson-commutative subalgebra. Then $\pazocal F$ is called \textit{complete} if $\dim \diff \F(x) = n$ almost everywhere, where
	 $\diff \F(x) = \{ \diff f(x), f \in \F\} \subset \T^*_x M$.
	 %, if it is complete on an everywhere dense subset.
\end{definition}
Let $\F  \subset \pazocal O(M)$ be a complete Poisson-commutative subalgebra. Consider an arbitrary $H \in \F$ and the corresponding Hamiltonian vector field
$$
	\mathrm{X}_H = \omega^{-1}\diff H.
$$
Then all functions in $\F$ are pairwise commuting integrals of $\mathrm{X}_H$, and $\mathrm{X}_H$ is completely integrable in the Liouville sense. So, formally, an integrable system is a complete commutative subalgebra $ \F$ with a distinguished Hamiltonian $H \in  \F$. However, the choice of $H \in  \F$ is not important to us, so we do not distinguish between integrable systems and complete commutative subalgebras.

% this reason, we will not distinguish between integrable systems and complete commutative subalgebras.

%
%Let $\mathcal F \subset \Cont^\infty(M^{2n})$ be a complete commutative subalgebra. Consider an arbitrary $H \in \F$ and the corresponding Hamiltonian vector field
%$$
%	\mathrm{X}_H = \omega^{-1}\diff H.
%$$
%Then all functions in $\F$ are pairwise commuting integrals of $\mathrm{X}_H$, and $\mathrm{X}_H$
%is a \textit{completely integrable Hamiltonian system}. So, an integrable system can be understood as a complete commutative subalgebra $\mathcal F \subset \Cont^\infty(M^{2n})$ with a distinguished element $H \in \mathcal F$. However, a particular choice of $H \in \mathcal F$ is not important to us. For
% this reason, we will not distinguish between integrable systems and complete commutative subalgebras.
%
%\begin{remark}
%Note that as a vector space, $\F$ may be infinite-dimensional.
%\end{remark}
%
%Consider an integrable system $\F$. Then the common level sets $\{\F = \const\}$ define a \textit{singular Lagrangian fibration} on $M^{2n}$ associated with $\F$.

\begin{definition}
	A point $x \in M^{2n}$ is called \textit{singular} for $\pazocal F$ if $\dim \diff \pazocal F(x) < n$. The number $\dim \diff \F(x)$ is called \textit{the rank} of a singular point $x$. The number $n - \dim \diff \F(x)$ is called \textit{the corank} of a singular point $x$.

\end{definition}
%
%A regular fiber of a singular Lagrangian fibration  is a fiber which does not contain singular points. By the Arnold-Liouville theorem, all compact regular fibers of a singular Lagrangian fibration are tori, and the dynamics on these tori is quasi-periodic.  However, the most interesting solutions of an integrable system, such as fixed points and stable periodic trajectories, belong to singular fibers. That is why it is important to study singularities of Lagrangian fibrations.
%\par
%As it usually happens in singularity theory, is is not realistic to describe all possible singularities, so one should start with studying the most generic of them. The most generic singularities of an integrable system are the \textit{non-degenerate} ones defined below. Details can be found in \cite{BolOsh, Eliasson, AL1}.

Let $x \in M^{2n}$ be a singular point of $\F$. Then there exists $H \in \F$ such that $\diff H(x) = 0$ and thus $\mathrm X_H = 0$. For such $H$, we can consider the linearization of the vector field $\mathrm{X}_H$ at the point $x$. This is a linear operator $A_{H}: \T_x M \to \T_x M$.  Let
$$
A_{\F} = \{ A_{H} \mid H \in \F, \diff H(x) = 0\}.
$$
As it is easy to see, $A_{\F}$ is a commutative subalgebra of $\sP(\T_{x}M, \omega)$. \par
Now consider the space $$W = \{\mathrm X_H(x), H \in \F\} \subset \T_{x}M.$$ Since the flows $\mathrm X_H$ where $H \in \F$ pairwise commute, the space $W$ is isotropic with respect to $\omega$. Let $W^\bot$ be the orthogonal complement to $W$ with respect to $\omega$. Then $W^{\bot} / W$ is symplectic. Furthermore, each operator $A_H \in A_{\F}$ vanishes on $W$, so it induces an operator $\pazocal A_H$ on $W^{\bot} / W$. In this way, we can reduce the commutative subalgebra $A_{\F} \subset \sP(\T_{x}M, \omega)$ to a commutative subalgebra  $\pazocal A_{\F} \subset \sP(W^{\bot} / W, \omega)$.
%Since $\F$ is commutative, $W$ is isotropic and all operators belonging to $A_{\F}$ vanish on $W$.  Consider the skew-orthogonal complement to $W$ with respect to $\omega$, i.e. the subspace
%$$W^\bot = \{ \xi \in \T_x M ~|~  \omega (\xi, W)=0\}.$$  Obviously, $W\subset W^\bot$, and $W^\bot$ is invariant under $A_{\F} $. Consequently, we can consider  elements of $A_{\F}$ as operators on $W^{\bot} / W$. Since $W$ is isotropic, the quotient $W^{\bot}/W$ carries a natural symplectic structure induced by $\omega$, and  $A_{\F} $ is a commutative subalgebra in $\sP(W^{\bot} / W, \omega)$.

\begin{definition}\label{nd}
	A singular point $x$ is called \textit{non-degenerate} if $\pazocal A_{\F}$ is a Cartan subalgebra in $\sP(W^{\bot} / W, \omega)$.
\end{definition}
In the complex case, all Cartan subalgebras are conjugate to each other. In the real case, Cartan subalgebras were classified in \cite{Williamson}.\par\pagebreak[4]
If $\h\subset \sP(2m, \R)$ is a Cartan subalgebra, then eigenvalues of any $A \in \h$ have the form
\begin{align*}
&\pm \lambda_{1}i, \dots, \pm \lambda_{e}i,\\
 &\pm \mu_{1}, \dots, \pm \mu_{h},\\
  &\pm \alpha_{1} \pm \beta_{1}i, \dots, \pm \alpha_{f} \pm \beta_{f}i,
\end{align*}
where $e + h + 2f = m$. The triple $(e,h,f)$ is the same for any regular $A \in \h$ and is called the \textit{type} of the Cartan subalgebra $\h$. Two Cartan subalgebras of $ \sP(2m, \R)$ are conjugate to each other if and only if they are of the same type.
\begin{definition}
	The \textit{type} of a non-degenerate singular point $x$ is the type of the associated Cartan subalgebra $\pazocal A_{\F} \subset \sP(W^{\bot} / W, \omega)$.
\end{definition}
For every non-degenerate singular point $x$ of rank $r$, the following equality holds:
$$
e + h + 2f  + r = n.
$$
The numbers $e,h,f$ are called the numbers of elliptic, hyperbolic, and focus-focus components respectively. %
% Let us formulate the Eliasson theorem on the linearization of a Lagrangian fibration in the neighbourhood of a non-degenerate singular point. Define the following standard singularities.
% \begin{enumerate}
% \item The fibration given by the function $p^2 + q^2$ in the neighbourhood of the origin in $(\R^2, \diff p \wedge \diff q)$ is called an {\it elliptic} singularity.
%  \item The fibration given by the function $pq$ in the neighbourhood of the origin in  $(\R^2, \diff p \wedge \diff q)$ is called a {\it hyperbolic} singularity.
%       \item The fibration given by the commuting functions  $p_1q_1 + p_2q_2, p_1q_2 - q_1p_2$ in the neighbourhood of the origin in  $(\R^4, \diff p_1 \wedge \diff q_1 + \diff p_2 \wedge \diff q_2)$ is called a {\it focus-focus} singularity.
% \end{enumerate}
%
\begin{theorem}[Vey \cite{vey}]
\label{EliassonThm}
	Let $\F$ be a real analytic\footnote{There also exist $\Cont^\infty$ and equivariant $\Cont^\infty$ versions of Theorem  \ref{EliassonThm}, see \cite{Eliasson, Miranda, MZ}.} integrable system\footnote{Our formulation of Theorem \ref{EliassonThm} is slightly different from the standard one. The latter assumes that $\F$ has dimension $n$ as a vector space. However, it is easy to show that these formulations are equivalent.}  and let $x$ be its non-degenerate singular point of rank $r$ and type $(e,h,f)$. Then there exist a Darboux chart $p_1, q_1,\dots, p_n, q_n$ centered at $x$ such that each $H \in \F$ can be written as
	$$
	H = H(f_1, \dots, f_n)
	$$
where 
\begin{align*}
f_i = \left[\begin{aligned}
&p_i^2 + q_i^2 \quad &\mbox{for } 1 \leq i \leq e,\\
&p_iq_i \quad &\mbox{for }  e+1 \leq i \leq e + h,\\
&p_iq_i + p_{i+1}q_{i+1} \quad &\mbox{for } i = e+h+1, e+h +3, \dots, e+h +2f-1,\\
&p_{i-1}q_i - p_{i}q_{i-1}  \quad &\mbox{for } i = e+h+2, e+h +4, \dots, e+h +2f, \\
&p_i \quad &\mbox{for } i > e+h + 2f.\\
\end{aligned}\right.
\end{align*}
Furthermore, there exist $H_1, \dots H_n \in \F$ such that $\det\left(\diffFXi{H_i}{f_j}(0)\right) \neq 0$.
	% Then the associated Lagrangian fibration is locally fiberwise symplectomorphic to the direct product of $k_e$ elliptic, $k_h$ hyperbolic, and $k_f$ focus-focus singularities, multiplied by a trivial non-singular fibration $\R^r \times \R^r$.
\end{theorem}
The geometric meaning of Theorem \ref{EliassonThm} is the following. Near a non-degenerate singular point $x$, the singular Lagrangian fibration $\{\F = \const\}$ is locally symplectomorphic to a product of the following standard fibrations:
 \begin{enumerate}
 \item elliptic fibration which is given by the function $p^2 + q^2$ in the neighbourhood of the origin in $(\R^2, \diff p \wedge \diff q)$;
  \item hyperbolic fibration which is given by the function $pq$ in the neighbourhood of the origin in $(\R^2, \diff p \wedge \diff q)$;
       \item focus-focus fibration which is given by the commuting functions  $p_1q_1 + p_2q_2, p_1q_2 - q_1p_2$ in the neighbourhood of the origin in  $(\R^4, \diff p_1 \wedge \diff q_1 + \diff p_2 \wedge \diff q_2)$;% is called a {\it focus-focus} singularity.
        \item non-singular fibration which is given by the function $p$ in the neighbourhood of the origin in $(\R^2, \diff p \wedge \diff q)$.
 \end{enumerate}
The dynamics in the neighborhood of a non-degenerate singular point can also be easily described. In particular, for a generic Hamiltonian $H \in \F$, the qualitative picture of the dynamics of $\mathrm{X}_H$ in the neighborhood of a non-degenerate singular point is  determined by the rank and type of this point. 

\par\smallskip
In the complex case, we have the following.\begin{theorem}
%\label{EliassonThm}
	Let $\F$ be a holomorphic integrable system  and let $x$ be its non-degenerate singular point of rank $r$. Then there exist a Darboux chart $p_1, q_1,\dots, p_n, q_n$ centered at $x$  such that each $H \in \F$ can be written as
	$$
	H = H(f_1, \dots, f_n)
	$$
where 
\begin{align*}
f_i = \left[\begin{aligned}
&p_iq_i \quad &\mbox{for }  i \leq n-r,\\
&p_i \quad &\mbox{for } i > n - r.\\
\end{aligned}\right.
\end{align*}
	% Then the associated Lagrangian fibration is locally fiberwise symplectomorphic to the direct product of $k_e$ elliptic, $k_h$ hyperbolic, and $k_f$ focus-focus singularities, multiplied by a trivial non-singular fibration $\R^r \times \R^r$.
\end{theorem}
%These two theorems can be interpreted as follows: if $x$ is a non-degenerate singular point, then all Hamiltonians $H \in \F$ can be simultaneously brought to the Birkhoff normal form in the vicinity of $x$, and this normal form is convergent.\par
%For a generic Hamiltonian $H \in \F$, the dynamics of $\mathrm{X}_H$ in the neighborhood of a non-degenerate singular point is completely determined by the rank and type of this point.\par
Now, if $M$ is a Poisson manifold, and $x \in M$, then there exists a unique symplectic leaf $O \subset M$ passing through $x$. This allows to transfer all definitions and statements of this section to Poisson manifolds. \par
The following lemma is useful for proving non-degeneracy in the Poisson setting.
\begin{lemma}\label{NDC}
Let $M$ be a Poisson manifold, and let $O \subset M$ be a generic symplectic leaf. Further, assume that $\F$ is a subspace of $\pazocal O(M)$ such that $\F\mid_O$ is an integrable system. Let $x \in O$ be a point of rank $k$ for $\F\mid_O$, and let
$$
V_x = \{ H \in \F \mid \mathrm{X}_H(x) = 0\}.
$$
Assume that there exist linearly independent $\phi_1, \dots, \phi_k \in V_x^*$ and non-zero $\eps_1^\pm, \dots, \eps_{k}^\pm \in \T^*_x M$ such that
$$
A_H^* \eps_i^\pm = \pm \phi_i(H)\eps_i^\pm
$$
for each $H \in V_x$.
Then:
\begin{longenum}
\item The space $W^\bot / W$ is spanned by $w_1^\pm, \dots, w_k^\pm$ such that
$$
\pazocal A_H w_i^\pm = \pm \phi_i(H)w_i^\pm
$$
for each $H \in V_x$.
\item The singular point $x$ is non-degenerate.
\item  In the real case, the type of $x$ is $(e,h,f)$ where $e$ is the number of pure imaginary $\phi_i$'s, $h$ is the number of real $\phi_i$'s, and $f$ is the number of pairs of complex conjugate $\phi_i$'s.
\end{longenum}

\end{lemma}
\begin{proof}
Assume that  $H \in V_x$.
Let $P \colon \T^*_x M \to \T_x O$ be the mapping defined by the Poisson tensor. Following \cite{SBSn}, we claim that the following diagram commutes:
\begin{align*}
       				 \begin{CD}
           			 	\T^*_x M @> A_H^* >> \T^*_x M\\
            				@VV P V  @VV P V \\
          				 \T_x O @> A_{H} >>\T_xO
       			 \end{CD}
   			 \end{align*}
%Indeed, let $\xi \in \T^*_x M$. Take a function $f$ such that $\diff f(x) = \xi$. Then for each function $g$ we have
%$$
%\langle \diff g , P(A_H^*(\diff f)) \rangle =- \langle A_H^*(\diff f), \mathrm X_g \rangle = -  \langle \diff f, \mathrm [X_H, X_g] \rangle
%$$
Therefore, if we take $e_i^\pm = P\eps_i^\pm$, then
$$
A_H e_i^\pm = \pm \phi_i(H)e_i^\pm.
$$
Let us show that $e_i^\pm \neq 0$. Indeed, if  $e_i^\pm = 0$, then $\eps_i^\pm \in \Ker P$. However, from regularity of the symplectic leaf $O$, we conclude that $A_H^* \mid_{\Ker P} = 0$ (see \cite{SBSn}), so $\eps_i^\pm \notin \Ker P$.\par
Now, note that since all operators $A_H$ vanish on the space $W$, we have $A_H(\T_xO) \subset W^\bot$, so $e_i^\pm \in W^\bot$ and $e_i^\pm \notin W$. Let $\pi$ be the projection $W^\bot \to W^\bot / W$. If we set $w_i^\pm = \pi(e_i^\pm)$, then $w_i^\pm \neq 0$, and
$$
\pazocal A_H w_i^\pm = \pm \phi_i(H)w_i^\pm.
$$
By dimension argument, $w_i^\pm$ span $W^\bot/W$, and operators $\pazocal A_H$ span a Cartan subalgebra in $\sP(W^{\bot} / W, \omega)$, q.e.d.
%Now, it suffices to show that $e_i^\pm \in W^\bot$. 
%%First, let us show that $A_H(\T_xM) \subset \T_xO$. Let $\xi \in \T_xM$. Extend $\xi$ to a vector field defined in the vicinity of $x$. We have
%$$
%A_H(\xi) = [\mathrm{X}_H, \xi](x).
%$$
%Let $f$ be a Casimir function defined in the neighborhood of $x$. We have
%$$
%\langle \diff f, [\mathrm{X}_H, \xi] \rangle(x) = \mathrm{X}_H(\xi(f)) (x)- \xi(\mathrm{X}_H(f))(x) = \mathrm{X}_H(\xi(f))(x) - \xi(\{H,f\})(x) =  0
%$$
%where we used that $\mathrm X_H(x) = 0$ and that $f$ is a Casimir function. Since $O$ is a generic symplectic leaf, its tangent space coincides with the annihilator of the space spanned by differentials of local Casimir functions, and $A_H(\xi) \in \T_xO$. Further, all operators $A_H$ vanish on the space $W$, so $A_H(\T_xO) \subset W^\bot$, and $A_H^2(\T_xM) \subset W^\bot$, which implies that $e_i^\pm \in W^\bot$, q.e.d.
\end{proof}
\nopagebreak
\subsection{Nodal curves and non-degenerate singularities}\label{ncns}
The space $\mathcal L_m^J (\gl(n,\Complex))$ carries an $m+1$-dimensional family of compatible Poisson structures, and the flows \eqref{loopI} are Hamiltonian with respect to each of these structures \cite{ReimanRev}. Each of these Poisson structures has rank $mn(n-1)$ almost everywhere. At some points the rank drops, however it is not difficult to show that for each point $L \in \mathcal L_m^J (\gl(n,\Complex))$, there exists a Poisson structure which has a maximal rank at this point. Therefore, for each point $L \in \mathcal L_m^J (\gl(n,\Complex))$, we can find a symplectic leaf of dimension $mn(n-1)$ passing through the point $L$. In what follows, we consider only such symplectic leaves.
\par\smallskip 
Let $\pazocal F = \{H_\psi\} $ be the integrable system constructed in the introduction. The following statement follows from Theorem \ref{thm1}.
\begin{theorem}\label{rkFormula}
Assume that $C$ is a nodal curve, and let $L \in \pazocal S_C$. Let also $O$ be a maximal dimension symplectic leaf passing through the point $L$. 
Then the rank of the point $L$ for the system $\F\mid_{O}$ is equal to $$\rank L =  \frac{mn(n-1)}{2}  - |K(L)|,$$ so that
$$
\corank L = |K(L)|.
$$
In particular, $L$ is singular for the system $\F\mid_{O}$ if and only if $K(L) \neq \emptyset$.

\end{theorem}
\begin{consequence}\label{singsing1}
If $C$ is a nodal curve, then 
$$
\rank  \pazocal S_C = \min\nolimits_{L \in  \pazocal S_C} \rank L =  \frac{mn(n-1)}{2}  - |\Sing (C)|,
$$
i.e. the corank is equal to the number of nodes.
\end{consequence}
\begin{consequence}\label{singsing}
If $C$ is a nodal curve, then $\pazocal S_C$ is singular, i.e. it contains at least one singular point.
\end{consequence}
%Apparently, Corollary \eqref{singsing} is true for all singular curves. In the case $m=1$ this is proved in \cite{Brailov}.
 
 The \textit{bifurcation diagram} $\mathcal B$ is the set of curves $C \in \mathcal C_{spec}$ such that $\pazocal S_C$ is singular. The \textit{discriminant} of the spectral curve $\mathcal D$ is the set of singular curves $C \in \mathcal C_{spec}$. Since for non-singular $C$ the fiber $\pazocal S_C$ is also non-singular, we have $\mathcal B \subset \mathcal D$. Since nodal curves are dense in $\mathcal D$, Corollary \ref{singsing} implies that we actually have $\overline{\mathcal B} = \mathcal D$ (as it is not difficult to see, $\mathcal D$ is closed). Apparently, $\mathcal B = \mathcal D$, i.e. $\pazocal S_C$ is singular if and only if $C$ is singular. For $m=1$, this is proved in \cite{Brailov}. For ``restricted systems'' discussed at the end of the introduction,  this result is not true \cite{Konyaev}. In particular, if $n$ is odd and we restrict $\F$ to  $\mathcal L_m^J(\so(n))$, then the spectral curve is \textit{always} singular. \par
The following theorem states that if the spectral curve $C$ is nodal, then all singular points on $\pazocal S_C$ are non-degenerate.

\begin{theorem}\label{NDT}
Assume that $C$ is a nodal curve, and that $O$ is a generic symplectic leaf passing through the point $L$.  Assume that $L \in \pazocal S_C$ is singular for the system $\F\mid_{O}$. Then 

\begin{enumerate} 
\item The singular point $L$ is non-degenerate.
\item In the real case, the type of $L$ is $(e,h,f)$ where $e$ is the number of acnodes in $K(L)$, $h$ is the number of crunodes in $K(L)$, and $f$ is one half the number of nodes in $K(L)$ which do not lie in the real part of $C$.
\end{enumerate}

%Then the rank of the point $L$ for the system $\F\mid_{O_q}$ is equal to $$\dim \diff\F\mid_{O_q} =  \frac{mn(n+1)}{2}  - |K(L)|.$$
%In particular, $L$ is singular for the system $\F\mid_{O_q}$ if and only if 
%$$ \frac{mn(n+1)}{2}  - |K(L)| < \frac{1}{2}\dim O_q,$$
%If the symplectic leaf $O_q$ is generic, then $L$ is singular if and only if $K(L) \neq \emptyset$.
%
\end{theorem}

As an example, consider the case $m=1$ and $n=3$ already discusses in Section \ref{ex1}. The corresponding spectral curve is a cubic. Table \ref{table2} lists all possible types of real nodal cubics and corresponding singularities. The column ``rank'' shows the minimal rank of singularities on $\pazocal S_C$. The column ``type'' shows the type of these minimal rank singular points. Note that the case ``degenerate quadric + line'' is only possible if $J$ has two complex eigenvalues, and the case ``three straight lines'' is only possible when all eigenvalues of $J$ are real.
\begin{table}
\centerline{\begin{tabular}{c|c|c|c}
Curve & Example & Rank & Type\\
\Xhline{2\arrayrulewidth}
\begin{tabular}{@{}c@{}}Irreducible cubic  \\  with an acnode \end{tabular}& $\vphantom{\displaystyle\int}\lambda^2(\mu - 3) - (\mu -1 )(\mu -2)^2 = 0$ & 2 & (1,0,0) \\
\hline 
\begin{tabular}{@{}c@{}}Irreducible cubic  \\  with a crunode \end{tabular}& $\vphantom{\displaystyle\int}\lambda^2(\mu - 3) + (\mu -1 )(\mu -2)^2 = 0$ & 2 & (0,1,0) \\
\hline 
\begin{tabular}{@{}c@{}}Quadric + line with two \\  real points in common\end{tabular} & $\vphantom{\displaystyle\int}(\lambda^2 + \mu^2 - 1)(\lambda - \mu) = 0$ & 1 & (0,2,0) \\
\hline 
\begin{tabular}{@{}c@{}}Quadric + line with no \\   real points in common\end{tabular} & $\vphantom{\displaystyle\int}(\lambda^2 + \mu^2 - 1)(\lambda - \mu + 2) = 0$ & 1 & (0,0,1) \\
\hline
\begin{tabular}{@{}c@{}}Degenerate quadric + line \\ in general position\end{tabular}  & $\vphantom{\displaystyle\int}(\lambda^2 + \mu^2)(\lambda - \mu + 2) = 0$ & 0 & (1,0,1) \\
\hline
\begin{tabular}{@{}c@{}}Three straight lines \\ in general position\end{tabular}  & $\vphantom{\displaystyle\int}(\lambda - \mu)(\lambda - 2\mu)(\lambda - 3\mu)  = 0$ & 0 & (0,3,0)
\\
\end{tabular}
 }
 \caption{Real nodal cubics and corresponding singularities of the $\gl(3)$ system}\label{table2}
 %\par\smallskip\smallskip
\end{table}
\par\smallskip
Apparently, the following converse result to Theorem \ref{NDT} is true: if $C$ is not nodal, then there exists at least one degenerate singular point in $\pazocal S_C$. We can prove this for some classes of curves, however this is beyond the scope of the present paper. %We note that in the real case, the converse to Theorem \ref{NDT}, in many cases, can be proved using the fact that non-degenerate singularities as well as they type survive under small perturbations. For example, take $m=1$ and $n = 3$, and consider a family of curves
 We note that if the curve $C$ is not nodal, then \textit{some} singular points in $\pazocal S_C$ may still be non-degenerate.\par\pagebreak[3]
Now, let us prove Corollary \ref{corIrCom}. Consider the set
$
\pazocal S_{C}^{(p)}
$
which consists of points of corank at least $r$. By Theorem \nolinebreak \ref{rkFormula}, we have
$$
\pazocal S_{C}^{(p)} = \bigsqcup_{|K| \geq p} \pazocal S_{C}^K.
$$
\begin{consequence}\label{rkcor}
Assume that $C \in \mathcal C_{spec}$ is a nodal curve. Then:

\begin{enumerate}
\item  The dimension of $\pazocal S_{C}^{(p)}$ equals $\frac{1}{2}mn(n-1) - p$.
%\item  The dimension of $\pazocal S_{C}^r$ equals $mn(n-1)/2 - r$.
\item If $p_2 > p_1$, then $\pazocal S_{C}^{(p_2)}$ lies in the closure of $\pazocal S_{C}^{(p_1)}$.
\item  The number of irreducible components of $\pazocal S_{C}^{(p)}$ is equal to the sum $\sum_{|K| = p}{|\Delta_K|}$ where $|\Delta_K|$ is the number of uniform multidegrees on $X'_K$.
%\item Let $K \subset \Sing C$. Then 
% Let $K_1,K_2 \subset \Sing C$. Then $\pazocal S_{C}^{K_1}$ lies in the closure of $\pazocal S_{C}^{K_2}$ if and only if $K_1 \supset K_2$. In other words, the set $\overline{\pazocal S_{C}^K}$ is the closure of $\pazocal S_{C}^K$.
%\item  The number of irreducible components of $\overline{\pazocal S_{C}^K}$ is equal to the number of uniform multidegrees on the curve $X'_K$.
\end{enumerate}
\end{consequence}
\begin{proof}
Assertion 1 follows from Theorem \ref{thm1}. Assertion 2 follows from the local description of non-degenerate singularities (Theorem \ref{EliassonThm}). Assertion 3 follows from Assertion \nolinebreak 2.
\end{proof}
\begin{proof}[Proof of Corollary \ref{corIrCom}]
Apply Corollary \ref{rkcor} for $r = 0$.
\end{proof}

The proof of Theorem \ref{NDT} is based on explicit formulae for eigenvalues of operators $\pazocal A_H$, $H \in \F,$ which are given below.
Assume that $L \in \pazocal S_{C}^K$, and let $\phi \in \Complex[\mu, \lambda^{-1}]$ be such that the vector field \eqref{loopI} vanishes at the point $L$. Then, by Theorem \ref{thm1}, we have
\begin{align}\label{freqVanish}
\sum_{P: \,\lambda(P) = \infty} \Res_{P}\, \phi\omega = 0
\end{align}
\nopagebreak for each differential $\omega$ regular on $X'_K$. \par \pagebreak[3]
Let $X$ be the non-singular compact model of $C$, and let $\pi \colon X \setminus \{\infty_1, \dots, \infty_n \} \to C$ be the normalization map. Assume that $$\Sing C = \{\underbrace{Q_1, \dots, Q_k}_{\mbox{in } K}, \underbrace{Q_{k+1}, \dots ,Q_l}_{\mbox{not in } K} \}.$$ Let $ \pi^{-1}(Q_i) = \{Q_i^+, Q_i^- \}$. Then regular differentials on $X'_K$ can be described as follows. These are differentials $\omega$ on $X$ which may have simple poles at points $Q_{k+1}^+,Q_{k+1}^-, \dots, Q_l^+,  Q_l^-$, and $\infty_1, \dots, \infty_n$, are holomorphic outside these points, and
$$
\Res_{Q_i^+} \omega + \Res_{Q_i^-} \omega = 0 \quad \forall \, i > k, \quad \sum_{i=1}^n \Res_{\infty_i} \omega = 0.
$$
Let $j \leq k$, and let us consider a differential $\omega_j$ on $X$ with the following properties: it may have simple poles at points $Q_j^\pm, Q_{k+1}^\pm, \dots, Q_l^\pm, \infty_1, \dots, \infty_n$, it is holomorphic outside these points, and
$$
\Res_{Q_j^\pm} \omega_j  = \pm 1, \quad \Res_{Q_i^+} \omega_j + \Res_{Q_i^-} \omega_j = 0 \quad \forall \, i > k, \quad \sum_{i=1}^n \Res_{\infty_i} \omega_j = 0.
$$
%
%
%%Recall that a differential $\omega$ on $X$ is called regular on  $X'_K$ if all poles of $\omega$ are simple, and for each $Q \in X'_K$, we have
%%$$
%%\sum_{P: \pi(P) = Q} \Res_P \,\omega = 0
%%$$
%%where $\pi \colon X \to X'_K$ is the normalization map.
%%
%%%Let $X$ be the non-singular compact model of $C$. Recall that we can identify regular differentials on $X'_K$ with differentials on $X$ with the following property: for each $Q \in X'_K$, we have
%%%$$
%%%\sum_{P: \pi(P) = Q} \Res_P \,\omega = 0
%%%$$
%%%where $\pi$ is the normalization map $\pi \colon X \to X'_K$. 
%%Let us also consider the curve $X'_\emptyset $ obtained from $C$ by adding points at infinity and  identifying $\infty_1 \sim \dots \sim \infty_n$. Let $\pi_\emptyset \colon X \to X'_\emptyset$ be the normalization map.
%%Let also $K = \{ Q_1, \dots, Q_n\}$, and let $\pi_\emptyset^{-1}(Q_i) = \{ P_i^+, P_i^-\}$. Consider a differential $\omega_i$ such that it is regular on $ X'_\emptyset$, and
%%$$
%%\Res_{P_j^\pm}\, \omega_i = \pm \delta_{ij}.
%%$$
Obviously, the differential $\omega_j$ is well-defined up to a differential which is regular on $X'_K$. So, by \eqref{freqVanish}, the numbers
\begin{align}\label{eigen}
\nu_j(\phi) = \sum_{P: \,\lambda(P) = \infty} \Res_{P}\, \phi\omega_j
\end{align}
are well-defined for each $\phi$ such that \eqref{loopI} vanishes at the point $L$. 
%For each vector field \eqref{loopI} vanishing at the point $L$, let $B_\phi \colon \T_L\mathcal L_m^J (\gl(n,\Complex)) \to \T_L\mathcal L_m^J (\gl(n,\Complex))$ be the linearization of this vector field at $L$. Let also $W$ spanned by the vectors \eqref{loopI} at $L$, and let $W^\bot$ be its orthogonal complement with respect to the symplectic structure on the symplectic leaf of the Poisson bracket $\{\,,\}_q$.
\begin{theorem}\label{evf}
Assume that $C$ is nodal curve, and that $K \subset \Sing C$. Let $L \in \pazocal S_{C}^K$.
Then the space $W^\bot / W$ (see Section \ref{nsis}) is spanned by the vectors $w_1^\pm, \dots, w_{k}^\pm$, and
for each $H_\psi \in \F$ such that the corresponding vector field \eqref{loopI} vanishes at the point $L$, we have
 $$
\pazocal A_{H_\psi} \,w_j^\pm = \pm\nu_j(\phi) w_j^\pm
$$
where $\nu_j$ is given by \eqref{eigen}, and $\phi = \diffFXi{\psi}{\mu}$.
\end{theorem}
Note that formulas \eqref{velocity} for the velocity vector on the Jacobian, and \eqref{eigen} for the eigenvalues of a linearized flow are, in essence, the same.
It is not difficult to see that this actually \textit{should} be so: when we approach a fixed point of a quasi-periodic flow, frequencies of the flow tend to the eigenvalues of its linearization at the fixed point.%, since eigenvalues are limiting values of frequencies.
%: eigenvalues of a quasi-periodic flow linearized at a fixed point are limits of the frequencies of the flow.

%Let us demonstrate that this \textit{should} be so by considering a simple example.\par Let $m = 1$, and let $C$ be the union of $n$ straight lines $l_1, \dots, l_n$ in general position, as in Section \ref{ex1}. The the space of regular forms on $X'_\emptyset$ has a natural basis $\{\omega_{ij}, 1 \leq i < j \leq n\}$. The form $\omega_{ij}$ has a residue $+1$ at the point of $l_i$ at which $l_i$ intersect $l_j$, and a residue $-1$ at the point on $l_j$ at which $l_j$ intersect $l_i$.
%Apart from this, $\omega_{ij}$ has a residue $-1$ at $\infty_i \in l_i$, and a residue $+1$ at $\infty_j \in l_j$. At all other points, $\omega_{ij}$  is holomorphic.\par The period lattice reads $2\pi \mathrm i\Z^{\frac{1}{2}n(n-1)}$, and the generalized Jacobian reads $$\Jac(X'_\emptyset) = \Complex^{\frac{1}{2}n(n-1)} / 2\pi \mathrm i\Z^{\frac{1}{2}n(n-1)}.$$
%For each $\phi \in \Complex[\mu, \lambda^{-1}]$, the equations of motion on the Jacobian read
%$$
%\dot \omega_{ij} = \sum_{P: \,\lambda(P) = \infty} \Res_{P}\, \phi\omega_{ij}.
%$$
%%Let us also choose $\phi_1, \dots, \phi_{n(n-1)/2} \in \Complex[\mu, \lambda^{-1}]$
%
%%The period lattice is
%$$
%\{2\pi k_{12} \mathrm i
%$$ 
\par
See Section \ref{ncnsproof} for the proof of Theorem \ref{NDT} and Theorem \ref{evf}.

%Each regular differential on $X$

%Recall that $X_K$ is the curve which is obtained from $C$ by adding points at infinity and blowing up at the points of $K$, and $X'_K$ is the curve obtained from $X_{K}$ by identifying $\infty_1 \sim \dots \sim \infty_n$. Taking $K = \emptyset$, we obtain the curve $X'_\emptyset$. Let $L \in  \pazocal S_C$, and let $K(L) = \{ P_1, \dots, P_k\}$. Then we can find regular differentials on $X_\emptyset$ such that

%Let $X_\emptyset $ be the curve obtained from $C$ by adding points $\infty_1, \dots, \infty_n$, and let $X'_\emptyset$ be the curve obtained from $X_\emptyset$ by identifying $\infty_1 \sim \dots \sim \infty_n$.

\subsection{Proof of Theorems \ref{NDT} and \ref{evf}}\label{ncnsproof}
Assume that the right-hand side of the equation \eqref{loopI} vanishes. By Theorem \ref{thm1}, this means that $\phi \in \Omega^1(X,\Sigma')^\bot$ where $
\Omega^1(X,\Sigma')^\bot 
%= \{ \phi \in \Complex[\mu,\lambda^{-1}] : \langle \phi, \omega \rangle_\infty = 0 \,\, \forall \,\, \omega \in \Omega^1(X,\Sigma') \}
$
is the left radical of the form $\langle\,,\rangle_\infty$ given by \eqref{HMAP}. 
%The tangent space  $\T_L L_m^J (\gl(n,\Complex))$ can be naturally identified with the space
%\begin{equation*}\mathcal L_{m-1} (\gl(n,\Complex))= \left\{  \sum_{i=0}^{m-1} L_i\lambda^i \mid  L_i \in \gl(n,\Complex) \right\}. %\subset \gl(n,\Complex) \otimes \Complex[\lambda]
%\end{equation*} 
Equation \nolinebreak \eqref{loopI} can be written as
\begin{align}\label{loopI2}
\diffXp{t}L= [L,A_\phi(L)]
\end{align}
where $A_\phi$ is a map $ \mathcal L_m^J (\gl(n,\Complex)) \to  \mathcal L_m^J (\gl(n,\Complex))$. The linearization of  \eqref{loopI2} is the operator $$B_\phi \colon \T_L L_m^J (\gl(n,\Complex)) \to \T_L L_m^J (\gl(n,\Complex)) $$ given by
$
B_\phi(Y) = [Y,A_\phi(L)] + [L,\diff A_\phi(Y)].
$\par
 Let us consider a map 
$$
R \colon \gl(n,\Complex) \times \Complex \to \T^*_L \mathcal L_m^J (\gl(n,\Complex)) 
$$
given by 
$
\langle R(A, a), Y \rangle = \Tr AY(a)
$
where the tangent space  $\T_L L_m^J (\gl(n,\Complex))$ is identified with the space
\begin{equation*}\mathcal L_{m-1} (\gl(n,\Complex))= \left\{  \sum_{i=0}^{m-1} L_i\lambda^i \mid  L_i \in \gl(n,\Complex) \right\}. %\subset \gl(n,\Complex) \otimes \Complex[\lambda]
\end{equation*} 
We have
\begin{align*}
\langle B_\phi^*(R(A, a)), Y \rangle = \Tr A[Y(a),&A_\phi(L)(a)] + \Tr A[L(a),\diff A_\phi(Y)(a)] = \\ &= \Tr [A_\phi(L)(a),A]Y(a) + \Tr [A,L(a)]\diff A_\phi(Y)(a).
\end{align*}
Assuming that $A$ is such that $  [A,L(a)]= 0$, we have
\begin{align}\label{coFormula}
B_\phi^*(R(A, a)) = R([A_\phi(L)(a),A],a).
\end{align}
At the same time, since $  [A,L(a)]= 0$ and $[L,A_\phi(L)]$ = 0, we have
$$
[[A_\phi(L)(a),A], L(a)] = [[A_\phi(L)(a),L(a)],A] + [A_\phi(L)(a),[A,L(a)]] = 0.,
$$
so the subspace $R(\mathfrak C(L(a)), a) \subset  \T^*_L L_m^J (\gl(n,\Complex)) $, where $\mathfrak C(L(a))$ is the centralizer of $L(a)$,  is invariant with respect to the operator $B_\phi^*$.\par
Further, let $h = h_\alpha \colon X \to \CP^{n-1}$ be the mapping constructed in Section \ref{cons}. This map satisfies the equation
$
L(\lambda)h = \mu h.
$
In a similar way, we construct a mapping $\xi \colon X \to \CP^{n-1}$ such that
$
L(\lambda)^*\xi = \mu \xi.
$
where $L(\lambda)^*$ is the adjoint operator (the transposed matrix).
Assume that $K(L) = \{Q_1, \dots, Q_k\}$, and let $\pi^{-1}(Q_i) = Q_i^\pm$.  Let also $a_i = \lambda(Q_i^\pm)$. Then
$$
 h(Q_i^+) \otimes \xi(Q_i^-)\in \mathfrak C(L(a_i)), \quad   h(Q_i^-) \otimes \xi(Q_i^+)  \in \mathfrak C(L(a_i)).
$$
Further, since $[L,A_\phi(L)] = 0$, there exists a meromorphic function $\nu$ on $X$ such that
$$
A_\phi(L)h = \nu h, \quad A_\phi(L)^*\xi = \nu \xi.
$$
Let
$$
\eps_i^+ = R(h(Q_i^+) \otimes \xi(Q_i^-), a_i), \quad \eps_i^- = R(h(Q_i^-) \otimes \xi(Q_i^+), a_i).
$$
Using \eqref{coFormula}, we have
$$
B_\phi^*\eps_i^+ =  (\nu(Q_i^+) - \nu(Q_i^-))\eps_i^+, \quad B_\phi^*\eps_i^- =  (\nu(Q_i^-) - \nu(Q_i^+))\eps_i^-.
$$

Let $\omega_i$ be a differential on $X$ with the following properties:\nopagebreak
\begin{enumerate}
\item it may have simple poles at points of $\mathrm{supp}(\Sigma')$ and $Q_i^\pm$;
\item it is holomorphic outside these points;
\item for each $\pazocal P \in \Sigma'$, we have
$$
\sum_{P \in \pazocal P}\Res_P\, \omega_i = 0;
$$
\item $\Res_{Q_i^\pm}\, \omega_i = \pm 1$.
\end{enumerate}
Clearly, such a differential exists and is unique modulo a $\Sigma'$-regular differential.
\begin{statement}
We have
$$
\nu(Q_i^-) - \nu(Q_i^+) =  \sum_{j=1}^n \Res_{\infty_j}\, \phi\omega_i.
$$
\end{statement}
\begin{proof}
We have 
\begin{align*}
\nu(Q_i^-) - \nu(Q_i^+) &= -\Res_{Q_i^+}\, \nu\omega_i - \Res_{Q_i^-}\, \nu\omega_i =  \\ &=\sum_{j=1}^n \Res_{\infty_j}\, \nu\omega_i + \sum_{j=1}^{|\Sigma|}\left(\Res_{P_j^+}\, \nu\omega_i + \Res_{P_j^-}\, \nu\omega_i\right). \end{align*}
As it is easy to see, we have $\nu \in \mathcal M(X,\Sigma)$, so the latter summand vanishes. At the same time, we have
$$
\phi(L,\lambda^{-1})_+\,h = \nu h, \quad \phi(L,\lambda^{-1})h = \phi(\mu, \lambda^{-1}) h,
$$
so
$$
\phi(L,\lambda^{-1})_-\,h = (\phi(\mu, \lambda^{-1}) - \nu) h,
$$
which implies that $\mathrm{ord}_{\infty_j} \phi(\mu, \lambda^{-1}) - \nu \geq 1$. Therefore,
$$
\nu(Q_i^-) - \nu(Q_i^+)  = \sum_{j=1}^n \Res_{\infty_j}\, \nu\omega_i =  \sum_{j=1}^n \Res_{\infty_j}\, \phi\omega_i,
$$
q.e.d.
\end{proof}
We conclude that there exist non-zero $\eps_1^\pm,\dots, \eps_k^\pm \in  \T^*_L L_m^J (\gl(n,\Complex))$ such that for each $\phi \in \Omega^1(X,\Sigma')^\bot$, we have
\begin{align}\label{eigenstar}
B_\phi^*\eps_i^\pm = \mp\left( \sum_{j=1}^n \Res_{\infty_j}\, \phi\omega_i \right)\eps_i^\pm.
\end{align}
%Dualizing, we get $2k$ elements $e_1^\pm,\dots, e_k^\pm \in  \T_L L_m^J (\gl(n,\Complex))$ such that for each $\phi \in \Omega^1(X,\Sigma')^\bot$, we have
%$$
%B_\phi e_i^\pm = \mp\left( \sum_{j=1}^n \Res_{\infty_j}\, \phi\omega_i \right)e_i^\pm.
%$$

Let $$\Sigma'' = \Sigma \cup \{\{ Q_1^+, Q_1^-\}. \dots, \{ Q_k^+, Q_k^-\}\},$$ and let us extend the pairing $ \langle \,, \rangle_\infty$ defined by \eqref{HMAP} to a pairing  
$$
\langle \,,  \rangle_\infty \colon \Complex[\mu,\lambda^{-1}] \times  \Omega^1(X,\Sigma'') \to \Complex
$$
by the same formula \eqref{HMAP}. The same argument as in Proposition \ref{completeness} shows that the right radical of the extended pairing is trivial, which implies that the right radical of the pairing
$$
\langle \,,  \rangle_\infty \colon  \Omega^1(X,\Sigma')^\bot \times \left( \Omega^1(X,\Sigma'') /  \Omega^1(X,\Sigma') \right) \to \Complex.
$$
is also trivial. Since the space $  \Omega^1(X,\Sigma'') /  \Omega^1(X,\Sigma') $ is spanned by $\omega_1, \dots, \omega_k$, we conclude that the functionals $\phi \mapsto \sum \Res_{\infty_j}\, \phi\omega_i$ are linearly independent. Now, Theorem \ref{evf} and the first assertion of Theorem \ref{NDT} follow from \eqref{eigenstar} and Lemma \ref{NDC}.

To prove the second assertion of Theorem \ref{NDT}, consider the anti-holomorphic involution $\tau \colon X \to X$ induced by the involution $(\lambda,\mu) \to (\bar \lambda, \bar \mu)$ on the spectral curve. As it is easy to see, for each point $P \in X$ and each meromorphic differential $\omega$, the following formula holds:
\begin{align}\label{conjres}
\overline{\vphantom{\tau^*}\Res_{ P}\, \omega} = \Res_{\tau(P)}\, \overline{\tau^*\omega}.
\end{align}
Consider three cases. First, assume that $Q_i$ is an acnode. Then $\tau$ swaps $Q_i^+$ and $Q_i^-$. Using formula \eqref{conjres}, we conclude that $ \overline{\tau^*\omega_i} = -\omega_i$ modulo a $\Sigma'$-regular differential, so
$$
\overline{ \sum \Res_{\infty_j}\, \phi\omega_i } = \sum \Res_{\,\overline \infty_j}\, \overline{\tau^*(\phi\omega_i)} = -\sum \Res_{\infty_j}\, \phi\omega_i
$$
where we used that $\overline{\tau^*\phi }= \phi$ and that for each $j$ there exists $k$ such that $\overline \infty_j = \infty_k$.
We conclude that the eigenvalues of $B_\phi^*$ corresponding to eigenvectors $e_i^\pm$ are pure imaginary. Analogously, if $Q_i$ is a crunode, then $\tau(Q_i^+) = Q_i^+$, and  the eigenvalues of $B_\phi^*$ corresponding to eigenvectors $e_i^\pm$ are real. Finally, if $Q_i$ and $Q_{i+1}$ are complex conjugate nodes, we get a quadruple of complex eigenvalues, q.e.d.

\bibliographystyle{plain}
\bibliography{Lax}
\end{document}